\documentclass[journal]{IEEEtran}
\usepackage{amsmath,amsfonts}
\usepackage{algorithmic}
\usepackage{array}
\usepackage{textcomp}
\usepackage{stfloats}
\usepackage{url}
\usepackage{verbatim}
\usepackage{graphicx}
\usepackage{cite}
\usepackage{hyperref}
\usepackage{booktabs}
\usepackage{diagbox}
\allowdisplaybreaks[4]

\usepackage[ruled]{algorithm2e}
\usepackage{subfigure}
\usepackage{float}
\usepackage{amssymb,amsfonts}
\usepackage{amsmath}
\usepackage{amsthm}
\theoremstyle{definition}

\newtheorem{proposition}{Proposition}
\newtheorem{remark}{Remark}
\usepackage{color}
\usepackage{graphicx}
\usepackage{multirow}
\usepackage{epsfig}
\usepackage{epstopdf} 
\usepackage{extarrows}
\usepackage{tikz}
\usetikzlibrary{arrows.meta}

\begin{document}
\title{Blind MIMO Semantic Communication via \\Parallel Variational Diffusion: \\A Completely Pilot-Free Approach}

\author{Hao Jiang,
Xiaojun Yuan, \textit{Senior Member}, \textit{IEEE}, Yinuo Huang, and Qinghua Guo, \textit{Senior Member}, \textit{IEEE}
\thanks{

H. Jiang, X. Yuan, and Y. Huang are with the National Key Laboratory of Wireless Communications, University of Electronic Science and Technology of China, Chengdu 611731, China (email: jh@std.uestc.edu.cn; xjyuan@uestc.edu.cn; ynhuang@std.uestc.edu.cn). Q. Guo is with the School of Electrical, Computer and Telecommunications Engineering, University of Wollongong, Wollongong, NSW 2522, Australia (email: qguo@uow.edu.au). The work in this paper has been partially presented in IEEE/CIC International Conference on Communications in China (ICCC) 2025 \cite{jiang2025blind}.
}
}


\maketitle

\begin{abstract}

In this paper, we propose a novel blind multi-input multi-output (MIMO) semantic communication (SC) framework named Blind-MIMOSC that consists of a deep joint source-channel coding (DJSCC) transmitter and a diffusion-based blind receiver. 
The DJSCC transmitter aims to compress and map the source data into the transmitted signal by exploiting the structural characteristics of the source data, while the diffusion-based blind receiver employs a parallel variational diffusion (PVD) model to simultaneously recover the channel and the source data from the received signal without using any pilots. The PVD model leverages two pre-trained score networks to characterize the prior information of the channel and the source data, operating in a plug-and-play manner during inference. This design allows only the affected network to be retrained when channel conditions or source datasets change, avoiding the complicated full-network retraining required by end-to-end methods.
This work presents the first fully pilot-free solution for joint channel estimation and source recovery in block-fading MIMO systems.
Extensive experiments show that Blind-MIMOSC with PVD achieves superior channel and source recovery accuracy compared to state-of-the-art approaches, with drastically reduced channel bandwidth ratio.

\end{abstract}


\section{Introduction}

With the evolution of wireless communications towards the sixth-generation (6G) era, meeting the demands for massive data transmission and ultra-high spectral efficiency has emerged as a critical challenge in communication system design. By deploying multiple antennas at both transmitter and receiver, the multi-input multi-output (MIMO) technology exploits spatial diversity and multiplexing capabilities to substantially enhance the spectral efficiency of the system \cite{larsson2014massive}. In coherent MIMO systems, accurate acquisition of channel state information (CSI) is crucial because the quality of CSI directly impacts the performance of signal detection and decoding \cite{tse2005fundamentals}. As MIMO systems scale up with more and more antennas \cite{rusek2012scaling}, the number of channel parameters to be estimated increases significantly, rendering accurate CSI acquisition increasingly difficult and overhead-intensive. Specifically, since the channel coherence time is limited physically, the rising cost for CSI acquisition proportionally reduces the resource available for data transmission, thereby decreasing the spectral efficiency.

Existing MIMO communication techniques mainly fall into two main categories, namely, the pilot‑based approaches and the blind detection approaches. For the pilot-based approaches \cite{hassibi2003how,coldrey2007trainingbased,chi2011training,yuan2018fundamental}, dedicated pilot symbols are inserted into each transmission block (typically with the duration shorter than the channel coherence time). At the receiver, the channel is first estimated based on the known pilots, and the estimated channel is then used for subsequent signal detection and decoding. The amount of time-frequency resource necessarily used for pilot transmission (referred to as pilot overhead) increases linearly with the input size of the MIMO channel, which may significantly lower the spectral efficiency of the system. 
To reduce pilot overhead, compressed sensing has been used to reduce the number of required pilots by exploiting the channel sparsity in a certain transformed domain, such as the angular domain \cite{bajwa2010compressed,kuai2019structured}. 
Specifically, the channel is assumed to have a sparse representation on a discrete grid in the considered transformed domain. However, the true path parameters rarely coincide exactly with the predetermined grid points in practice, causing the renowned energy leakage problem that severely degrades the channel estimation accuracy \cite{he2019super}. To mitigate the grid mismatch, the super‐resolution approach proposed in \cite{he2019super} treats the grid points as learnable parameters and jointly optimizes both the sparse coefficients and the continuous grid points, thereby significantly enhancing channel estimation accuracy.
More recently, with the emergence of deep neural networks (DNNs), DNN-based methods have been developed for pilot design and channel estimation \cite{soltani2019deep,liu2024learning}. By learning the correlations of the channel across different transformed domains, these methods further reduce pilot overhead while enhancing channel estimation accuracy. Nevertheless, in all these pilot-based approaches, the pilot overhead is proportional to the number of transmit antennas.
As the number of MIMO antennas grows, the pilot overhead increases sharply, and the coherence time may be too short even to accommodate the necessary amount of pilots \cite{rusek2012scaling}.

To avoid the above difficulty, blind detection \cite{jiang2024generalized,zhang2018blind,dean2018blind,ding2019sparsity,yuan2023blind,schmid2025blind,liu2025blind} has been developed to estimate the channel and detect the transmitted signal simultaneously from the received signal without the aid of pilots, followed by channel and source decoding. The existing blind detection methods usually rely on bilinear factorization of the received signal to jointly estimate the channel and the transmitted signal. The main difficulty for blind detection is the non-uniqueness of the bilinear factorization of the received signal. Additional structural constraints are introduced to enforce the uniqueness of the bilinear factorization.
For example, the sparsity of channel coefficients (in a certain transformed domain) \cite{zhang2018blind,yuan2023blind} is exploited to ensure the uniqueness of sparse factorization up to permutation and phase ambiguities. These ambiguities exist because permuted and phase-shifted solutions are statistically indistinguishable under the routine assumption that the channel coefficients and the transmitted symbols are respectively independently and identically distributed (i.i.d.).
Alternatively, by introducing an additional sparsity constraint on the transmitted symbols \cite{ding2019sparsity}, the signal sparsity is exploited to ensure the uniqueness of sparse factorization up to the ambiguities.
These ambiguities can be further resolved by inserting additional reference symbols known to the receiver
into the transmitted signal \cite{zhang2018blind}. Yet, with this treatment, the blind detection schemes in \cite{zhang2018blind,dean2018blind,ding2019sparsity,yuan2023blind,schmid2025blind,liu2025blind} are not pilot-free in a strict sense. Moreover, due to the data processing inequality \cite{tse2005fundamentals},
the blind detection schemes in \cite{zhang2018blind,dean2018blind,ding2019sparsity,yuan2023blind,schmid2025blind,liu2025blind} 
generally suffer from performance loss due to the separation of the detection and decoding operations.

To break through the limitations of the existing pilot-based and blind detection approaches, we propose a novel blind MIMO semantic communication (Blind-MIMOSC) framework that enables joint recovery of the channel and the source data from the received signal without using any pilots. The proposed Blind-MIMOSC consists of a deep joint source-channel coding (DJSCC) transmitter and a diffusion-based blind receiver. The DJSCC transmitter aims to compress and map the source data into channel input symbols by exploiting the structural characteristics of the source data (such as the texture and contours of an image \cite{dong2015image} and the temporal correlations between frames in a video \cite{xue2021denoising}). The blind receiver simultaneously recovers the channel and the source data from the received signal without the aid of any pilots.

In the proposed transceiver design, we aim to minimize a variational deputy of the joint entropy of the channel and the source conditioned on the received signal, thereby reducing the uncertainty of both the channel and the source data given the received signal. For the transmitter design, since the CSI is unknown to our blind transceiver, we adopt established DJSCC codec training methodologies to train the DJSCC encoder without incorporating channel information.
For the blind receiver design, motivated by the recent advances in large artificial intelligence (AI) models such as score-based diffusion models \cite{song2021scorebased}, we propose a novel parallel variational diffusion (PVD) model, in which two pre-trained score networks are utilized to individually characterize the prior information of the channel and the source. 
Unlike conventional model-based methods that rely on explicit priors introduced by over-simplified assumptions (e.g., sparse signaling), the score networks can learn and exploit complex structural characteristics inherent in the channel and the source data, thereby significantly enhancing the recovery performance.

The proposed PVD is distinct from the conventional diffusion models in the following aspects. First, in PVD, two parallel reverse diffusion processes are employed to simultaneously recover both the channel and the source data. Second, PVD directly optimizes the variational means of the channel and the source data at each reverse diffusion step via gradient descent, instead of relying solely on heuristic stochastic sampling as in conventional diffusion models \cite{chung2023diffusion}. Third, the likelihood scores involved in the diffusion model are dependent on the nonlinear encoder and the bilinear channel model, which is difficult to calculate exactly. To achieve a more accurate approximation, we employ second-order score networks to take into account the errors of the channel and source estimates.


We conduct extensive numerical simulations to validate the performance of the proposed Blind-MIMOSC scheme. Our experimental setup includes both the Rayleigh fading channel model and the standardized 3GPP TR 38.901 CDL-C channel model.
The source data are images sampled from the FFHQ-256 validation set \cite{karras2019style}. We show that our Blind-MIMOSC with PVD successfully achieves joint channel-and-source recovery in a completely pilot-free manner, and significantly outperforms the state-of-the-art pilot-based DJSCC-MIMO algorithm \cite{wu2024deep} across various performance metrics.
Specifically, Blind-MIMOSC with PVD reduces the distortion and improves the perceptual quality of the recovered source data in terms of metrics such as the multi-scale structural similarity index measure (MS-SSIM) \cite{wang2003multiscale}, the deep image structure and texture similarity (DISTS) \cite{ding2020image}, and the learned perceptual image patch similarity (LPIPS) \cite{Zhang_2018_CVPR}.
The channel estimation accuracy of Blind-MIMOSC, quantified by the normalized mean-squared error (NMSE), is close to the oracle bound that assumes perfect knowledge of the transmitted signal in channel estimation.
Notably, these improvements are achieved with a drastically lower channel bandwidth ratio (CBR, defined as the ratio of the channel input symbol count to the source data dimension), which can be as low as 1/3 of the CBR of the pilot-based DJSCC-MIMO algorithm \cite{wu2024deep} in our considered simulation settings.

\subsection{Contributions}
The contributions of this work are summarized as follows:
\begin{itemize}
    \item We establish a novel blind MIMO semantic communication framework termed Blind-MIMOSC, which comprises a DJSCC transmitter and a diffusion-based blind receiver that, without the use of any pilots, simultaneously recovers both the channel and the source data from the received signal. Blind-MIMOSC is designed to minimize the joint entropy of the channel and source data conditioned on the received signal, which is in line with the information-maximization principle \cite{cover2012elements}. Interestingly, we show that both the pilot-based and blind detection approaches allow for a unified interpretation from the perspective of information maximization, except for the introduction of additional assumptions for simplified separate signal processing. In this regard, these existing approaches perform much worse than the proposed Blind-MIMOSC scheme since the latter conducts joint channel estimation and source decoding.

    \item For the blind receiver, we develop a novel PVD model for the simultaneous recovery of the channel and the source data. The proposed blind receiver employs individual pre-trained score networks to characterize the prior information of the channel and the source data, rather than imposing simplified assumptions such as the i.i.d. assumption \cite{zhang2018blind,dean2018blind,ding2019sparsity,yuan2023blind,schmid2025blind,liu2025blind}, thus achieving completely pilot-free joint channel-and-source recovery. Note that these two score networks are used in the inference process in a plug-and-play manner. As such, whenever the channel condition or the source data distribution varies, only the corresponding score network needs to be retrained, as in contrast to the complicated retraining of the whole neural network in end-to-end methods \cite{wu2024deep,cai2025end}.
    
    \item To the best of our knowledge, this work presents the \textit{first practical}, \textit{completely pilot-free} solution for joint channel estimation and source recovery in block-fading coherent MIMO systems.\footnote{Note that in noncoherent MIMO systems, pilot signals can be omitted since explicit CSI is not required for symbol detection.} 
    Extensive experiments validate the superior performance of the proposed approach for image transmission over different channel conditions. 
    Blind-MIMOSC with PVD significantly outperforms state-of-the-art pilot-based and blind schemes in terms of a variety of performance metrics, such as MS-SSIM \cite{wang2003multiscale}, DISTS \cite{ding2020image}, and LPIPS \cite{Zhang_2018_CVPR} of the recovered source data, and the NMSE of the channel estimation. Note that these improvements are achieved with a drastically lower CBR.

\end{itemize}

\subsection{Related Work}

In this subsection, we review prior work on the design of semantic communication systems.
With the advent of DNNs, end-to-end optimized DJSCC for data transmission has emerged as an active research area in semantic communications \cite{bourtsoulatze2019deep,yang2024swinjscc,wu2024deep}. Bourtsoulatze \emph{et al.} \cite{bourtsoulatze2019deep} proposed the convolutional autoencoder–based JSCC scheme for single-input single-output (SISO) scenarios, which outperforms classical separate source-channel coding schemes at relatively low signal-to-noise ratio (SNR). 
Later, Yang \emph{et al.} \cite{yang2024swinjscc} proposed SwinJSCC for SISO systems, which employs a Swin Transformer backbone to capture both local and global source features. 
Experimental results show that SwinJSCC achieves higher recovery accuracy than \cite{bourtsoulatze2019deep} with reduced end-to-end latency. More recently, Wu \emph{et al.} \cite{wu2024deep} developed the DJSCC-MIMO algorithm for MIMO semantic communication systems. By leveraging a Vision Transformer architecture, DJSCC-MIMO jointly exploits semantic features and channel conditions via self-attention, yielding improvements in source recovery accuracy and robustness to channel estimation errors across different SNR settings. 
However, among the schemes above, the variation of either the source data distribution or the channel distribution necessitates end-to-end retraining, resulting in a large training overhead for practical deployment. 

With recent advances in large generative AI models, score-based diffusion models \cite{song2019generative, song2021scorebased} have emerged as powerful tools for tasks that demand fine-grained control, such as audio and image synthesis, and are therefore well suited for source data recovery in semantic communication systems. In prior work, diffusion models have been adopted for semantic decoding \cite{grassucci2023generative,wang2025diffcom}. 
Specifically, diffusion posterior sampling (DPS) \cite{chung2023diffusion}, initially developed for general noisy inverse problems, has been extended to source data reconstruction in \cite{wang2025diffcom}. This scheme first performs pilot-based channel estimation and then recovers source data via DPS based on the received signal and the estimated channel. Experimental results demonstrate that this DPS-based approach offers superior source reconstruction quality in terms of a variety of metrics for SISO systems. 
Nevertheless, DPS-based approaches rely on heuristic stochastic sampling, and thus often encounter failures among the posterior samples due to the randomness, particularly in the blind MIMO semantic communication scenarios with unknown channel, nonlinear encoding, and channel mixing.
To address this limitation, our proposed PVD model directly optimizes the variational means of the channel and the source data through deterministic gradient descent across reverse diffusion steps. This deterministic approach significantly reduces sampling randomness, enabling more stable and consistent source data reconstructions compared to the DPS-based approaches.

\subsection{Organization and Notations}
The remainder of this paper is organized as follows. Section II introduces the system model, formulates the problem, and proposes the Blind-MIMOSC framework. Section III details the blind receiver design via the PVD model. Section IV introduces the extension of the Blind-MIMOSC framework and PVD model to the multi-user case. Section V presents numerical results and comparisons with state-of-the-art baselines. Finally, Section VI concludes the paper.

\textit{Notations}: Throughout this manuscript, we use standard mathematical notations and conventions. Specifically, we denote the set of complex numbers by $\mathbb{C}$, the set of real numbers by $\mathbb{R}$, and we employ regular lowercase letters, bold lowercase letters, and bold uppercase letters to represent scalars, vectors, and matrices, respectively. 
Transpose and conjugate transpose operations are denoted by $(\cdot)^\top$ and $(\cdot)^\mathrm{H}$, respectively. The Frobenius norm is denoted by $\Vert\cdot \Vert_F$, and $\propto$ denotes equality up to a constant multiplicative factor. We denote a multivariate normal distribution as $\mathcal{N}(\cdot; \boldsymbol{M}, \sigma^2)$, where $\boldsymbol{M}$ is the mean and all elements are mutually independent with equal variance $\sigma^2$. $\mathcal{CN}(\cdot; \boldsymbol{M}, \sigma^2)$ denotes a circularly symmetric complex normal distribution with mean $\boldsymbol{M}$ and scalar variance $\sigma^2$. We define $H(\boldsymbol{X}_1)$ for the entropy of a random variable $\boldsymbol{X}_1$, $I(\boldsymbol{X}_1;\boldsymbol{X}_2)$ for the mutual information of $\boldsymbol{X}_1$ and $\boldsymbol{X}_2$, $\operatorname{diag}[\boldsymbol{X}_1,\cdots,\boldsymbol{X}_N]$ for a bock-diagonal matrix with matrices $\left\{\boldsymbol{X}_n\right\}^{N}_{n=1}$ on its diagonal, and $\mathbb{E}_{p(x)}[\cdot]$ for the expectation over the distribution $p(x)$.

\section{Blind MIMO Semantic Communication Framework}\label{ProbForm}
\subsection{Channel Model}
Consider a block-fading MIMO system where a transmitter equipped with $N_t$ antennas communicates to a receiver with $N_r$ antennas. The MIMO channel remains static within each transmission block and undergoes independent fading across different blocks. Without loss of generality, we assume that each block can be further divided into $T$ time slots. The received signal at the receiver in the $k$-th transmission block is represented by
\begin{align}\label{systemmodel0}
    \boldsymbol{Y}_{k} &= \tilde{\boldsymbol{H}}_k\boldsymbol{X}_{k} + \boldsymbol{N}_{k},
\end{align}
where $k\in\{1, \dots, K\}$ with $K$ being the number of transmission blocks, $\tilde{\boldsymbol{H}}_k \in\mathbb{C}^{N_r\times N_t}$ is the channel between the transmitter and the receiver in the $k$-th transmission block, $\boldsymbol{X}_k\in \mathbb{C}^{N_t \times T}$ is the transmitted signal in the $k$-th transmission block, and $\boldsymbol{N}_k\in \mathbb{C}^{N_r \times T}$ is the corresponding additive white Gaussian noise (AWGN) matrix with elements drawn from the circularly symmetric complex Gaussian distribution $\mathcal{CN}(0, \sigma_n^2)$ with $\sigma_n^2$ being the noise power.
The received signal across all $K$ blocks can be expressed as
\begin{align}\label{systemmodel1}
    \boldsymbol{Y} ={}& [ \boldsymbol{Y}^{\top}_{1},\boldsymbol{Y}^{\top}_{2},\cdots ,\boldsymbol{Y}^{\top}_{K} ]^{\top}
    =\boldsymbol{H}_{0}\boldsymbol{X} + \boldsymbol{N},
\end{align}
where $\boldsymbol{H}_0 = \operatorname{diag}[\tilde{\boldsymbol{H}}_1, \tilde{\boldsymbol{H}}_2, \cdots, \tilde{\boldsymbol{H}}_K]\in \mathbb{C}^{N_rK \times N_tK}$ is the compound block-diagonal MIMO channel matrix, $\boldsymbol{X} = [ \boldsymbol{X}^{\top}_{1},\boldsymbol{X}^{\top}_{2},\cdots ,\boldsymbol{X}^{\top}_{K}]^{\top}\in\mathbb{C}^{N_tK\times T}$ is the transmitted signal, and $\boldsymbol{N} = [ \boldsymbol{N}^{\top}_{1},\boldsymbol{N}^{\top}_{2},\cdots ,\boldsymbol{N}^{\top}_{K} ]^{\top}\in \mathbb{C}^{N_rK \times T}$ is the AWGN across all the blocks.

\subsection{Transceiver Structure}\label{TransceiverFramework}
We consider a blind transceiver scheme over the block-fading MIMO channel in \eqref{systemmodel1}. The terminology ``blind'' means that: 1) the channel $\boldsymbol{H}_0$ is unknown to both the transmitter and the receiver; and 2) there is no pilot symbol inserted into the transmitted signal $\boldsymbol{X}$.\footnote{A pilot symbol is assumed to be known at the receiver, and can be used to facilitate channel estimation.} We start with the transmit side. Let $\boldsymbol{D}_0$ be the source data to be transmitted, which is typically a real-valued tensor whose dimensions depend on the specific data type (such as audio, images, and videos). 
A DJSCC encoder $f_{\boldsymbol{\gamma}}$, parameterized by $\boldsymbol{\gamma}$, maps source data $\boldsymbol{D}_0$ into the transmitted signal matrix as
\begin{align}\label{encoder}
    \boldsymbol{X} = f_{\boldsymbol{\gamma}}(\boldsymbol{D}_0)\in\mathbb{C}^{N_tK\times T},
\end{align}
where the encoder output is arranged into a matrix with dimension $N_tK \times T$ to conform with the MIMO configuration in \eqref{systemmodel1}. At the receive side, a diffusion-based blind receiver $g_{\boldsymbol{\varTheta}}$, parameterized by $\boldsymbol{\varTheta}$, directly recovers both the channel $\boldsymbol{H}_0$ and the source data $\boldsymbol{D}_0$ from the received signal $\boldsymbol{Y}$ without pilot assistance, i.e.,
\begin{align}\label{receiver}
    \{\hat{\boldsymbol{H}}_0, \hat{\boldsymbol{D}}_0\} = g_{\boldsymbol{\varTheta}}\left(\boldsymbol{Y}\right),
\end{align}
where $\hat{\boldsymbol{H}}_0$ and $\hat{\boldsymbol{D}}_0$ are the recovered channel and source data, respectively.





\subsection{Problem Formulation}
Based on \eqref{systemmodel1}-\eqref{encoder}, the joint distribution of the received signal, channel, transmitted signal, and source data can be written as
\begin{align}\label{jointcompute}
    &p\left(\boldsymbol{Y}, \boldsymbol{H}_0,\boldsymbol{X}, \boldsymbol{D}_0\right) \nonumber\\={}& p\left(\boldsymbol{Y}\vert \boldsymbol{H}_0,\boldsymbol{X}\right)p\left(\boldsymbol{H}_0\right)p_{\boldsymbol{\gamma}}\left(\boldsymbol{X}\vert \boldsymbol{D}_0\right)p\left(\boldsymbol{D}_0\right),
\end{align}
where $p(\boldsymbol{H}_0)$ and $p\left(\boldsymbol{D}_0\right)$ represent the \textit{a priori} distributions of the channel and the source data, respectively. The conditional probability density function (pdf) of $\boldsymbol{Y}$ given $\boldsymbol{H}_0$ and $\boldsymbol{X}$ is $p\left(\boldsymbol{Y}\vert \boldsymbol{H}_0,\boldsymbol{X}\right) = \mathcal{CN}\left(\boldsymbol{Y}; \boldsymbol{H}_0\boldsymbol{X}, \sigma^2_n\right)$ according to the channel model in \eqref{systemmodel1}. Note that the encoding process in \eqref{encoder} is deterministic, the conditional pdf of $\boldsymbol{X}$ given $\boldsymbol{D}_0$ is $p_{\boldsymbol{\gamma}}\left(\boldsymbol{X}\vert\boldsymbol{D}_0\right) = \delta\left(\boldsymbol{X}- f_{\boldsymbol{\gamma}}\left(\boldsymbol{D}_0\right)\right)$ with $\delta\left(\cdot\right)$ being the Dirac-delta distribution. 
With the joint distribution given in \eqref{jointcompute}, the probabilistic model of Blind-MIMOSC is given by
\begin{equation}
\begin{tikzpicture}[baseline=(Y.base), thick, >={Stealth[scale=0.8]}]
    \node (D0) at (0,0) {$\boldsymbol{D}_0$};
    \node (X0) at (2,0) {$\boldsymbol{X}$};
    \node (H0) at (2,1) {$\boldsymbol{H}_0$};
    \node (Y)  at (3,0.5) {$\boldsymbol{Y}$};
    \node (HhatDhat)  at (5.5,0.5) {$\{\hat{\boldsymbol{H}}_0, \hat{\boldsymbol{D}}_0\},$};
    
    \draw[->] (D0) -- node[above] {$f_{\boldsymbol{\gamma}}(\cdot)$} (X0);
    \draw[->] (X0) -- (Y);
    \draw[->] (H0) -- (Y);
    \draw[->, dashed] (Y) -- node[above] {$g_{\boldsymbol{\varTheta}}\left(\cdot\right)$} (HhatDhat);
\end{tikzpicture}
\end{equation}
where the dashed arrow from $\boldsymbol{Y}$ to the estimates $\hat{\boldsymbol{H}}_0$ and $\hat{\boldsymbol{D}}_0$ represents the inference
process of the blind receiver. 

Our Blind-MIMOSC follows the conditional entropy minimization \cite[Ch. 2.2]{cover2012elements}, which aims to minimize the uncertainty of both the channel $\boldsymbol{H}_0$ and the source data $\boldsymbol{D}_0$ given the received signal $\boldsymbol{Y}$, i.e.,
\begin{align}\label{entropy1}
    &H\left(\boldsymbol{H}_0,\boldsymbol{D}_0\vert  \boldsymbol{Y}\right)
    = \mathbb{E}_{p\left(\boldsymbol{Y}, \boldsymbol{H}_0,\boldsymbol{D}_0\right)}\left[-\ln p_{\boldsymbol{\gamma}}\left(\boldsymbol{H}_0,\boldsymbol{D}_0\vert  \boldsymbol{Y}\right)\right].
\end{align}
It is noteworthy that minimizing the conditional entropy $H\left(\boldsymbol{H}_0,\boldsymbol{D}_0\vert\boldsymbol{Y}\right)$ is equivalent to maximizing the mutual information between $\left\{\boldsymbol{H}_0, \boldsymbol{D}_0\right\}$ and $\boldsymbol{Y}$, denoted by $I\left(\boldsymbol{H}_0,\boldsymbol{D}_0;\boldsymbol{Y}\right)$. This equivalence stems from the fact of $I\left(\boldsymbol{H}_0,\boldsymbol{D}_0;\boldsymbol{Y}\right) = H\left(\boldsymbol{H}_0,\boldsymbol{D}_0\right) - H\left(\boldsymbol{H}_0,\boldsymbol{D}_0\vert\boldsymbol{Y}\right)$ \cite[Ch. 2.4]{cover2012elements}, where $H\left(\boldsymbol{H}_0, \boldsymbol{D}_0\right)$, the joint entropy of the channel and the source data, remains constant for a given system. Thus, minimizing the conditional entropy $H\left(\boldsymbol{H}_0, \boldsymbol{D}_0 \vert \boldsymbol{Y}\right)$ is equivalent to maximizing the mutual information $I\left(\boldsymbol{H}_0,\boldsymbol{D}_0;\boldsymbol{Y}\right)$.

The minimization of $H\left(\boldsymbol{H}_0,\boldsymbol{D}_0 \mid \boldsymbol{Y} \right)$ actually serves as a general objective for both blind and pilot-based approaches. In the pilot-based case, this conditional entropy can be upper-bounded by the sum of two loss terms, namely, two cross-entropy loss functions for training the channel estimator and the JSCC codec, respectively. Detailed discussions are provided in Appendix~\ref{pilotbased}.
As for the design of blind approaches, existing blind receiver designs typically follow a two-stage scheme comprising: 1) blind channel estimation and signal detection via a blind detector $g_{\boldsymbol{\theta}_1}$ parameterized by $\boldsymbol{\theta}_1$, and 2) subsequent JSCC decoding via a decoder $g_{\boldsymbol{\theta}_2}$ parameterized by $\boldsymbol{\theta}_2$. To see this, the conditional entropy is relaxed to the sum of two cross-entropy loss functions as
\begin{subequations}\label{BCAS}
\begin{align}
    &H\left(\boldsymbol{H}_0,\boldsymbol{D}_0\vert  \boldsymbol{Y}\right)\nonumber \\ 
    \leq{}& H\left(\boldsymbol{H}_0, \boldsymbol{X}, \boldsymbol{D}_0\vert\boldsymbol{Y}\right)\label{eq:step1}\\
    ={}& H\left(\boldsymbol{H}_0, \boldsymbol{X}\vert\boldsymbol{Y}\right) + H\left(\boldsymbol{D}_0\vert \boldsymbol{H}_0,\boldsymbol{X},\boldsymbol{Y}\right)\label{eq:step2}\\
    \leq{}& H\left(\boldsymbol{H}_0, \boldsymbol{X}\vert\boldsymbol{Y}\right) + H\left(\boldsymbol{D}_0\vert \boldsymbol{X}\right)\label{eq:step3}\\
    \leq{}& \underbrace{\mathbb{E}_{p\left(\boldsymbol{H}_0, \boldsymbol{X},\boldsymbol{Y}\right)}\left[-\ln q_{\boldsymbol{\theta}_1}\left(\boldsymbol{H}_0, \boldsymbol{X}\vert  \boldsymbol{Y}\right)\right]}_{\text{Loss function of the blind detector $g_{\boldsymbol{\theta}_1}$}}\nonumber\\
    &+\underbrace{\mathbb{E}_{p_{\boldsymbol{\gamma}}(\boldsymbol{D}_0,\boldsymbol{X})}[-\ln q_{\boldsymbol{\theta}_2}(\boldsymbol{D}_0\vert\boldsymbol{X})]}_{\text{Loss function of the JSCC codec $\{f_{\boldsymbol{\gamma}},g_{\boldsymbol{\theta}_2}\}$}}\label{eq:step7},
\end{align} 
\end{subequations}
where $q_{\boldsymbol{\theta}_1}\left(\boldsymbol{H}_0, \boldsymbol{X}\vert\boldsymbol{Y}\right)$ and $q_{\boldsymbol{\theta}_2}(\boldsymbol{D}_0\vert\boldsymbol{X})$ are the variational approximations of the true \textit{a posteriori} distributions $p\left(\boldsymbol{H}_0, \boldsymbol{X}\vert\boldsymbol{Y}\right)$ and $p_{\boldsymbol{\gamma}}(\boldsymbol{D}_0\vert\boldsymbol{X})$, respectively. To understand this derivation, consider the following explanations for each step. The inequality in \eqref{eq:step1} holds because adding the intermediate variable $\boldsymbol{X}$ (the transmitted signal) increases the joint conditional entropy. The equality in \eqref{eq:step2} follows from the chain rule of conditional entropy. The inequality in \eqref{eq:step3} follows from the property that dropping off condition variables $\boldsymbol{H}_0$ and $\boldsymbol{Y}$ increases the uncertainty of $\boldsymbol{D}_0$, i.e., $H\left(\boldsymbol{D}_0\vert \boldsymbol{X}\right) \geq H\left(\boldsymbol{D}_0\vert \boldsymbol{H}_0,\boldsymbol{X},\boldsymbol{Y}\right)$. 
Finally, the inequality in \eqref{eq:step7} introduces two cross-entropies that upper-bound the terms in \eqref{eq:step3}, which are loss functions for training 
the blind detector $g_{\boldsymbol{\theta}_1}$ and the JSCC codec $\{f_{\boldsymbol{\gamma}},g_{\boldsymbol{\theta}_2}\}$, respectively.
\begin{figure}[htb]
    \centering
    \subfigure[]{
        \centering
        \includegraphics[width=0.325\linewidth]{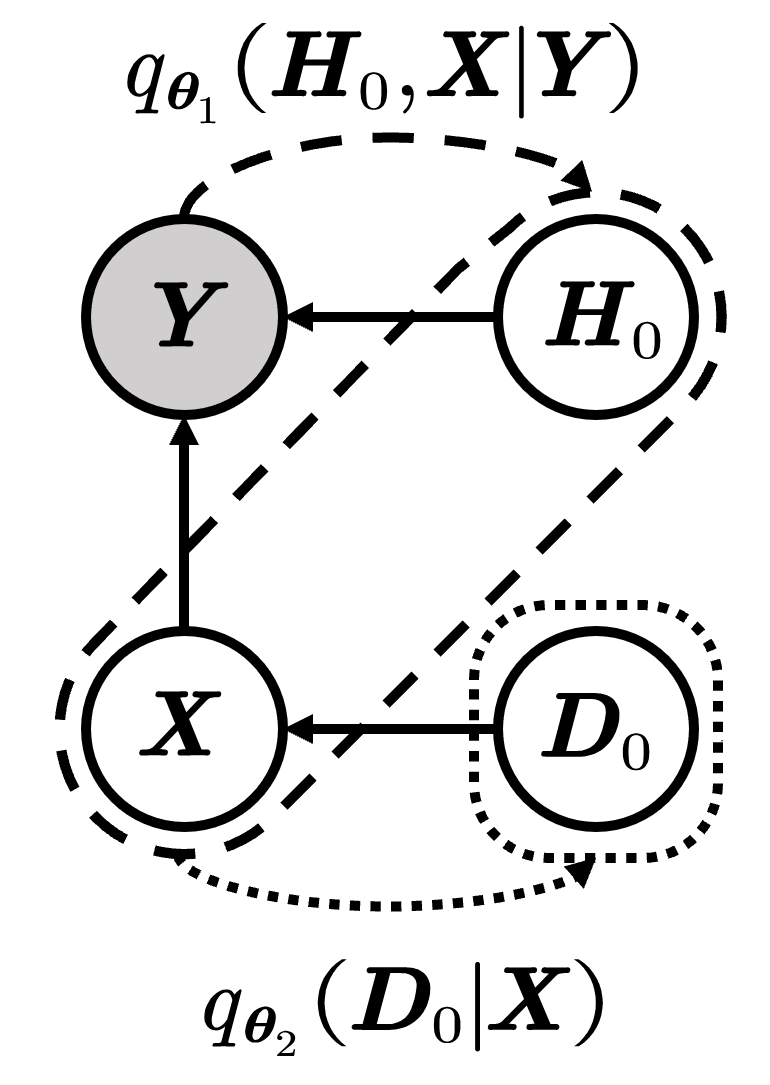}
    }
    \subfigure[]{
        \centering
        \includegraphics[width=0.325\linewidth]{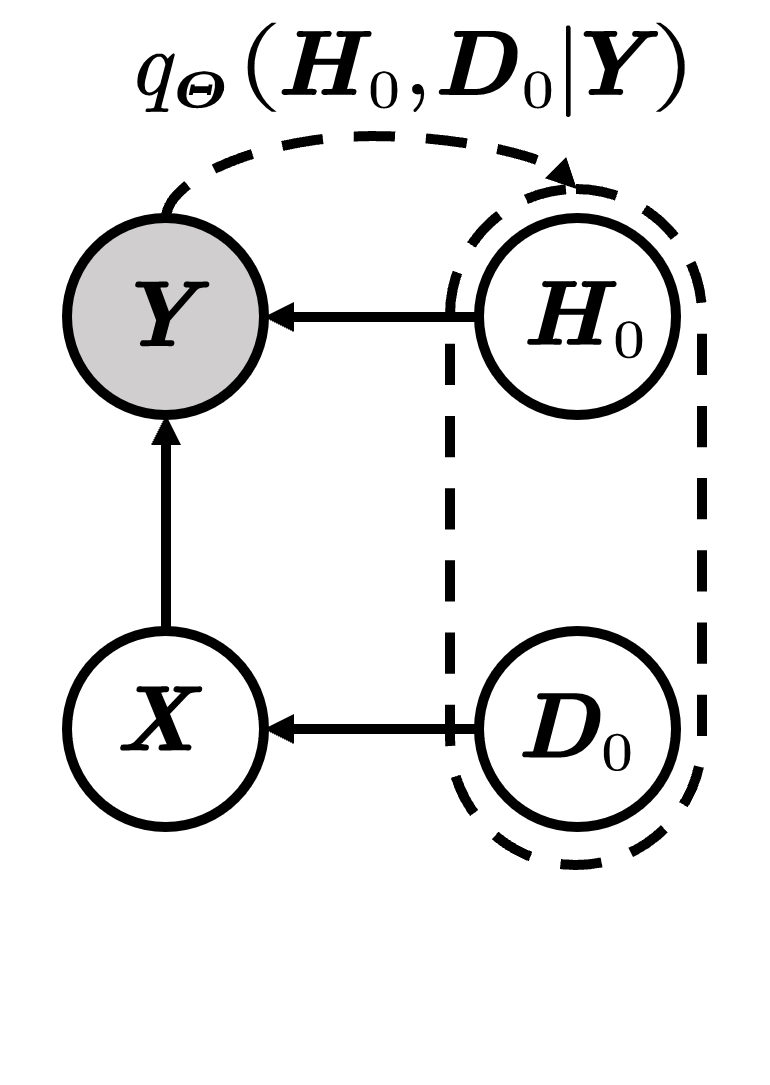}
    }
    \caption{Probabilistic graph of (a) blind channel estimation and signal detection followed by JSCC decoding; (b) blind channel-and-source recovery. Solid arrows denote the encoding and the MIMO transmission processes, and dashed arrows denote the variational inference process.
    }
    \label{inferenceDAG}
\end{figure}

From \eqref{BCAS}, the resulting receiver employs a two-stage inference scheme as shown in Fig. \ref{inferenceDAG} (a). Although these two-stage blind detection approaches reduce pilot overhead, they suffer from the following three drawbacks. 
First, since it is challenging to directly minimize \eqref{eq:step7}, existing blind detection approaches typically minimize the two loss functions in \eqref{eq:step7} separately. Due to the data processing inequality \cite{tse2005fundamentals} and the introduction of additional assumptions on the channel $\boldsymbol{H}_0$ and the signal $\boldsymbol{X}$ to facilitate blind detection, such as assuming that their elements are i.i.d. \cite{zhang2018blind,dean2018blind,ding2019sparsity,yuan2023blind,schmid2025blind,liu2025blind}, these two-stage blind detection approaches are suboptimal.
Second, in the first stage, estimating $\boldsymbol{H}_0$ and $\boldsymbol{X}$ via the bilinear factorization of the received signal $\boldsymbol{Y}$ suffers from non-unique factorization. 
Ding \emph{et al.} \cite{ding2019sparsity} introduced a sparsity constraint on the transmitted signal to ensure a unique sparse factorization up to permutation and phase ambiguities, as permuted and phase-shifted solutions are statistically unidentifiable under the aforementioned i.i.d. assumptions on $\boldsymbol{H}_0$ and $\boldsymbol{X}$.
However, from an information-theoretic perspective, sparse signaling is generally far from optimal in information maximization.
Third, to eliminate the remaining permutation and phase ambiguities, additional reference symbols known to the receiver must be inserted into $\boldsymbol{X}$ \cite{zhang2018blind,dean2018blind,ding2019sparsity,yuan2023blind,schmid2025blind,liu2025blind} for ambiguity calibration. These reference symbols prevent the blind detection schemes in \cite{zhang2018blind,dean2018blind,ding2019sparsity,yuan2023blind,schmid2025blind,liu2025blind} from being completely pilot-free.

To avoid the disadvantages of the above two-stage blind approach, we formulate a blind channel-and-source recovery problem to directly estimate the channel $\boldsymbol{H}_0$ and the source data $\boldsymbol{D}_0$ from the received signal $\boldsymbol{Y}$. 
The joint distribution of the received signal, channel, and source data in the conditional entropy \eqref{entropy1} is
\begin{align}
    p\left(\boldsymbol{Y}, \boldsymbol{H}_0, \boldsymbol{D}_0\right) ={}& \int p\left(\boldsymbol{Y}, \boldsymbol{H}_0, \boldsymbol{X}, \boldsymbol{D}_0\right) d\boldsymbol{X} \nonumber\\={}& p_{\boldsymbol{\gamma}}\left(\boldsymbol{Y}\vert \boldsymbol{H}_0, \boldsymbol{D}_0\right)p\left(\boldsymbol{H}_0\right)p\left(\boldsymbol{D}_0\right),
\end{align}
with $p_{\boldsymbol{\gamma}}\left(\boldsymbol{Y}\vert \boldsymbol{H}_0,\boldsymbol{D}_0\right) = \mathcal{CN}\left(\boldsymbol{Y}; \boldsymbol{H}_0f_{\boldsymbol{\gamma}}\left(\boldsymbol{D}_0\right), \sigma^2_n\right)$,
and the \textit{a posteriori} distribution is 
\begin{align}\label{postcompute0}
    p_{\boldsymbol{\gamma}}\left(\boldsymbol{H}_0,\boldsymbol{D}_0\vert\boldsymbol{Y}\right) \propto{}& p_{\boldsymbol{\gamma}}\left(\boldsymbol{Y}\vert \boldsymbol{H}_0, \boldsymbol{D}_0\right)p\left(\boldsymbol{H}_0\right)p\left(\boldsymbol{D}_0\right).
\end{align}
Note that \eqref{postcompute0} is generally intractable due to the unknown \textit{a priori} distributions of the channel and the source data. To avoid this, a widely used approach is to replace the \textit{a posteriori} distribution $p_{\boldsymbol{\gamma}}\left(\boldsymbol{H}_0,\boldsymbol{D}_0\vert\boldsymbol{Y}\right)$ with a variational distribution $q_{\boldsymbol{\varTheta}}\left(\boldsymbol{H}_0,\boldsymbol{D}_0\vert\boldsymbol{Y}\right)$ parameterized by $\boldsymbol{\varTheta}$ \cite{cai2025end}. 
This yields the cross-entropy that upper-bounds the conditional entropy as
\begin{align}
    H\left(\boldsymbol{H}_0,\boldsymbol{D}_0\vert  \boldsymbol{Y}\right)
    \leq\mathbb{E}_{p_{\boldsymbol{\gamma}}\left(\boldsymbol{Y},\boldsymbol{H}_0,\boldsymbol{D}_0\right)}\left[-\ln q_{\boldsymbol{\varTheta}}\left(\boldsymbol{H}_0,\boldsymbol{D}_0\vert \boldsymbol{Y}\right)\right] 
    \label{LVUB}.
\end{align}   
Therefore, the conditional entropy minimization is relaxed to the minimization of the cross-entropy:
\begin{subequations}\label{optimizeP1}
\begin{align}
    \min_{\boldsymbol{\gamma}, \boldsymbol{\varTheta}}& \quad\mathbb{E}_{p_{\boldsymbol{\gamma}}\left(\boldsymbol{Y},\boldsymbol{H}_0,\boldsymbol{D}_0\right)}\left[-\ln q_{\boldsymbol{\varTheta}}\left(\boldsymbol{H}_0,\boldsymbol{D}_0\vert \boldsymbol{Y}\right)\right],\\
    \text{s. t.}& \quad \frac{1}{N_tKT}\left\Vert  \boldsymbol{X}\right\Vert^2_F \leq P,\label{constraintX}
\end{align}
\end{subequations}
where the average power constraint $P$ is imposed on the transmitted signal. By solving the optimization problem in \eqref{optimizeP1}, our proposed Blind-MIMOSC framework is realized with the encoder $f_{\boldsymbol{\gamma}}$ and the blind receiver $g_{\boldsymbol{\varTheta}}$ which simultaneously recovers the channel and the source data without using pilots.


Solving \eqref{optimizeP1} poses the following serious challenges. First, the bilinear channel model in \eqref{systemmodel1}
complicates the problem as it may lead to a non-unique bilinear factorization, which hinders the blind receiver $g_{\boldsymbol{\varTheta}}$ from learning an effective mapping from the received signal $\boldsymbol{Y}$ directly to both the channel $\boldsymbol{H}_0$ and the source data $\boldsymbol{D}_0$.
Second, unlike the existing blind approaches that exploit the signal sparsity \cite{ding2019sparsity}, our blind channel-and-source recovery scheme can only rely on the structural characteristics of the source data, such as the texture and contours of an image \cite{dong2015image} and temporal correlations between frames in a video \cite{xue2021denoising}. However, these structures usually do not allow for explicit mathematical expressions, necessitating implicit learning through DNNs. 
In fact, there is so far no successful implementation of a completely pilot-free MIMO semantic communication system that can achieve reliable joint channel-and-source recovery prior to the work in this paper.

\subsection{Proposed Transceiver Design}
To efficiently solve \eqref{optimizeP1}, we take the following approach to sequentially train the encoder and the blind receiver. 

\subsubsection{Transmitter Design}

For the design of the DJSCC encoder in Blind-MIMOSC, we adopt established DJSCC codec training methodologies \cite{wu2024deep, yang2024swinjscc} to train a codec pair, denoted as $\{f_{\boldsymbol{\gamma}}, g_{\boldsymbol{\beta}}\}$, where $g_{\boldsymbol{\beta}}$ is the corresponding decoder of the encoder $f_{\boldsymbol{\gamma}}$. Since the channel is unknown to the transceiver, this training is conducted without incorporating channel information by solving the following optimization problem:
\begin{subequations}
\begin{align}
    \min_{\boldsymbol{\gamma}, \boldsymbol{\beta}}& \quad \left\Vert \boldsymbol{D}_0 - g_{\boldsymbol{\beta}}(f_{\boldsymbol{\gamma}}(\boldsymbol{D}_0)) \right\Vert^2_F,\\ \text{s. t.}& \quad \frac{1}{N_t K T} \left\Vert f_{\boldsymbol{\gamma}}(\boldsymbol{D}_0) \right\Vert^2_F \leq P,
\end{align}
\end{subequations}
where 
$P$ is the power constraint in \eqref{optimizeP1}. 
Notably, the decoder $g_{\boldsymbol{\beta}}$ is not utilized in Blind-MIMOSC. Given that the training of DJSCC codecs has been thoroughly explored in existing literature, we focus on the blind receiver design as follows.

\subsubsection{Blind Receiver Design}
For a given encoder $f_{\boldsymbol{\gamma}}$,
the optimization problem in \eqref{optimizeP1} is now equivalent to 
\begin{align}\label{optimizeP4}
    \min_{q_{\boldsymbol{\varTheta}}} \quad
    \mathbb{E}_{p\left(\boldsymbol{Y},\boldsymbol{H}_0,\boldsymbol{D}_0\right)}\left[-\ln q_{\boldsymbol{\varTheta}}\left(\boldsymbol{H}_0,\boldsymbol{D}_0\vert \boldsymbol{Y}\right)\right].
\end{align}  
The minimization can be achieved when $q_{\boldsymbol{\varTheta}}\left(\boldsymbol{H}_0,\boldsymbol{D}_0\vert \boldsymbol{Y}\right)=p_{\boldsymbol{\gamma}}\left(\boldsymbol{H}_0,\boldsymbol{D}_0\vert \boldsymbol{Y}\right)$.
Since the explicit \textit{a priori} distributions of the channel and the source data in \eqref{postcompute0} are unavailable, we approximate $p\left(\boldsymbol{H}_0\right)$ and $p\left(\boldsymbol{D}_0\right)$ by training unconditional score-based diffusion models as the generative priors, denoted by $q_{\boldsymbol{\theta}_H}\left(\boldsymbol{H}_0\right)$ and $q_{\boldsymbol{\theta}_D}\left(\boldsymbol{D}_0\right)$, respectively \cite{song2021scorebased}. Details of the score networks are described later in Section~\ref{Section3}. Consequently, the variational distribution is expressed as
\begin{align}
    q_{\boldsymbol{\varTheta}}\left(\boldsymbol{H}_0,\boldsymbol{D}_0\vert \boldsymbol{Y}\right) 
    \propto 
    p_{\boldsymbol{\gamma}}\left(\boldsymbol{Y}\vert \boldsymbol{H}_0,\boldsymbol{D}_0\right)q_{\boldsymbol{\theta}_H}\left(\boldsymbol{H}_0\right)q_{\boldsymbol{\theta}_D}\left(\boldsymbol{D}_0\right),
\end{align}
with $\boldsymbol{\varTheta} = \{\boldsymbol{\theta}_H, \boldsymbol{\theta}_D\}$.
The inference process of the proposed blind receiver is depicted in Fig. \ref{inferenceDAG} (b).
Notably, the proposed blind receiver architecture utilizes pre-trained components, i.e., the encoder $f_{\boldsymbol{\gamma}}$ and the unconditional score-based diffusion models, eliminating the need for resource-intensive joint retraining when source datasets or channel conditions change. 

\begin{remark}
To achieve unique bilinear factorization of the received signal to obtain the channel and the transmitted signal, traditional blind detection methods \cite{zhang2018blind,dean2018blind,ding2019sparsity,yuan2023blind,schmid2025blind,liu2025blind} typically rely on the use of reference symbols to resolve the permutation and phase ambiguities, as permuted and phase-shifted solutions are statistically unidentifiable under their i.i.d. assumptions on the elements of $\boldsymbol{H}_0$ and $\boldsymbol{X}$. In contrast, Blind-MIMOSC employs score-based diffusion models to characterize the complex prior information of $\boldsymbol{H}_0$ and $\boldsymbol{D}_0$, eliminating the need for additional i.i.d. assumptions and thereby avoiding the associated ambiguity issues. Consequently, the proposed Blind-MIMOSC can achieve joint recovery of the channel and signal source without requiring any pilot signals.
\end{remark}




\section{Blind Receiver Design via Parallel Variational Diffusion}\label{Section3}

We now consider the blind receiver design problem in \eqref{optimizeP4}. 
Directly approximating $p_{\boldsymbol{\gamma}}\left(\boldsymbol{H}_0,\boldsymbol{D}_0\vert \boldsymbol{Y}\right)$ is challenging due to the bilinear constraint in \eqref{systemmodel1} as well as the nonlinear encoder in \eqref{encoder}. To tackle this, we propose the score-based parallel variational diffusion (PVD) model. 
Specifically, we introduce latent variables of the score-based diffusion models, and upper-bound the objective in \eqref{optimizeP4} with the cross-entropy between the joint \textit{a posteriori} distribution (containing all latent variables) and its variational approximation. 
The parameters $\boldsymbol{\varTheta}$ are learned from the forward process of PVD.
For the inference process, the minimization of the upper bound is partitioned into tractable subproblems for each reverse diffusion step in PVD, where the variational means of both the channel and the source are gradually refined across reverse diffusion steps via gradient descent. Detailed descriptions of the model are provided in the following subsections.
\vspace{-0.1cm}
\begin{figure}[H]
    \centering
    \includegraphics[width=0.999\linewidth]{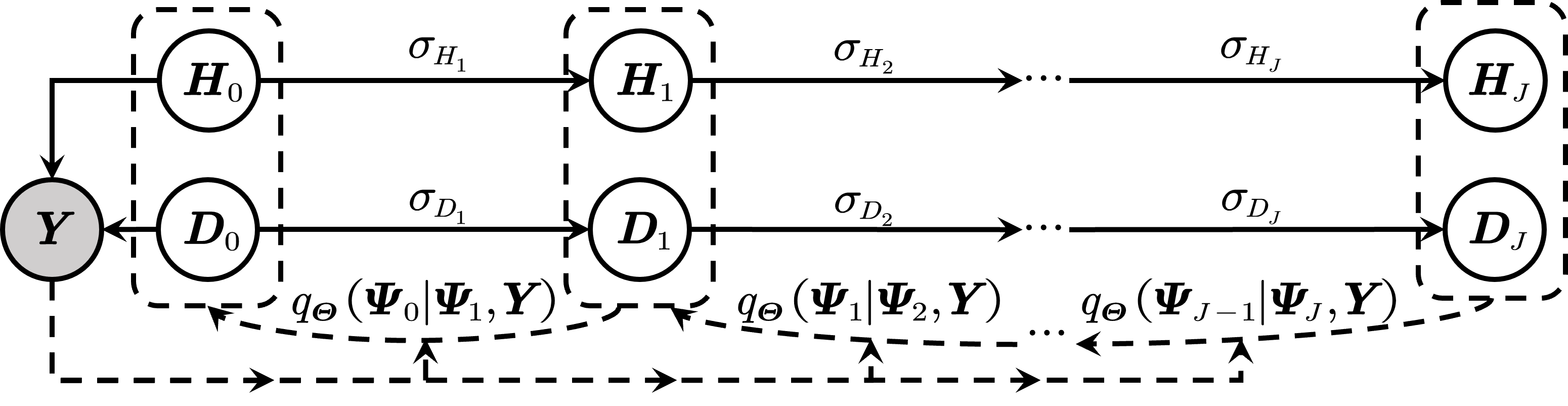}
    \vspace{-0.3cm}
    \caption{Probabilistic graph of the proposed blind receiver based on the parallel variational diffusion model. Solid arrows 
    denote the forward diffusion process of PVD. Dashed arrows denote the reverse diffusion process of PVD.}
    \label{BlindReceiver}
\end{figure}
\subsection{Problem Transformation}
We first introduce the forward diffusion process of PVD, where Gaussian noise is gradually added to the original data until the data are overwhelmed by the Gaussian noise. More specifically, the forward diffusion processes of the channel and the source are two independent Markov processes with the joint probability of their latent variables $\{\boldsymbol{H}_j\}_{j=0}^{J}$ and $\left\{\boldsymbol{D}_j\right\}_{j=0}^{J}$ given by
\begin{subequations}\label{diffpro}
\begin{align}
    p\left(\{\boldsymbol{H}_j\}_{j=0}^{J}\right) &= p\left(\boldsymbol{H}_{0}\right)\prod\nolimits_{j=0}^{J-1}p\left(\boldsymbol{H}_{j+1}\vert \boldsymbol{H}_{j}\right),\\
    p\left(\{\boldsymbol{D}_j\}_{j=0}^{J}\right) &= p\left(\boldsymbol{D}_0\right)\prod\nolimits_{j=0}^{J-1}p\left(\boldsymbol{D}_{j+1}\vert \boldsymbol{D}_j\right),
\end{align}
\end{subequations}
which are represented by the probabilistic graph in Fig. \ref{BlindReceiver}. The latent variables are given by
\begin{subequations}
\begin{align}
    \boldsymbol{H}_{j} &= \boldsymbol{H}_{0} + \sigma^{2}_{H_{j}}\boldsymbol{N}_{H_j}, \forall j\in\{0, 1, \cdots, J\}\\
    \boldsymbol{D}_j &= \boldsymbol{D}_0 + \sigma^{2}_{D_{j}}\boldsymbol{N}_{D_j}, \forall j\in\{0, 1, \cdots, J\}
\end{align}
\end{subequations}
where $\sigma^{2}_{H_{J}} > \cdots> \sigma^{2}_{H_{1}} > \sigma^{2}_{H_{0}} = 0$ and $\sigma^{2}_{D_{J}} > \cdots> \sigma^{2}_{D_{1}} > \sigma^{2}_{D_{0}} = 0$ are the pre-defined noise variances, $\boldsymbol{N}_{H_j}$ and $\boldsymbol{N}_{D_j}$ are AWGN matrices with elements drawn from $\mathcal{CN}(0, 1)$ and $\mathcal{N}(0, 1)$, respectively.
We thus have the conditional distributions:
\begin{subequations}\label{conditionaldis}
\begin{align}
    p\left(\boldsymbol{H}_{j+1}\vert \boldsymbol{H}_{j}\right) &=  \mathcal{CN}(\boldsymbol{H}_{j+1}; \boldsymbol{H}_{j}, \sigma^{2}_{H_{j+1}} - \sigma^{2}_{H_{j}}),\\
    p\left(\boldsymbol{D}_{j+1}\vert \boldsymbol{D}_j\right) &=  \mathcal{N}(\boldsymbol{D}_{j+1}; \boldsymbol{D}_j, \sigma^{2}_{D_{j+1}} - \sigma^{2}_{D_{j}}).
\end{align}
\end{subequations}
The parameters $\boldsymbol{\varTheta}$ of the pre-trained unconditional score networks are learned from the forward diffusion process via denoising score matching \cite{meng2021estimating} by solving
\begin{align}
    \min_{\boldsymbol{\theta}_H}&\quad\mathbb{E}_{p(\boldsymbol{H}_j)}[\|\boldsymbol{S}_{\boldsymbol{\theta}_H}(\boldsymbol{H}_j, \sigma_{H_j})-\nabla_{\boldsymbol{H}_j}\ln p(\boldsymbol{H}_j)\|^2_F \nonumber\\&\quad+ \|s_{\boldsymbol{\theta}_H}(\boldsymbol{H}_j, \sigma_{H_j})-\operatorname{tr}(\nabla^2_{\boldsymbol{H}_j}\ln p(\boldsymbol{H}_j))\|^2_F],\\
    \min_{\boldsymbol{\theta}_D}&\quad\mathbb{E}_{p(\boldsymbol{D}_j)}[\|\boldsymbol{S}_{\boldsymbol{\theta}_D}(\boldsymbol{D}_j, \sigma_{D_j})-\nabla_{\boldsymbol{D}_j}\ln p(\boldsymbol{D}_j)\|^2_F\nonumber\\&\quad+ \|s_{\boldsymbol{\theta}_D}(\boldsymbol{D}_j, \sigma_{D_j})-\operatorname{tr}(\nabla^2_{\boldsymbol{D}_j}\ln p(\boldsymbol{D}_j))\|^2_F].
\end{align}
The resulting score networks are the score of the generative priors. 
Specifically, the first-order score networks are
\begin{subequations}\label{firstorderscore}
    \begin{align}
        \boldsymbol{S}_{\boldsymbol{\theta}_H}(\boldsymbol{H}_j, \sigma_{H_j})&=\nabla_{\boldsymbol{H}_j}\ln q_{\boldsymbol{\theta}_H}\left(\boldsymbol{H}_j\right),\\
        \boldsymbol{S}_{\boldsymbol{\theta}_D}(\boldsymbol{D}_j, \sigma_{D_j})&=\nabla_{\boldsymbol{D}_j}\ln q_{\boldsymbol{\theta}_D}\left(\boldsymbol{D}_j\right),
    \end{align}
\end{subequations}
and the second-order trace score networks \cite{meng2021estimating} are
\begin{subequations}\label{secondorderscore}
    \begin{align}
        s_{\boldsymbol{\theta}_H}(\boldsymbol{H}_j, \sigma_{H_j})&=\operatorname{tr}(\nabla^2_{\boldsymbol{H}_j}\ln q_{\boldsymbol{\theta}_H}\left(\boldsymbol{H}_j\right)),\\
        s_{\boldsymbol{\theta}_D}(\boldsymbol{D}_j, \sigma_{D_j})&=\operatorname{tr}(\nabla^2_{\boldsymbol{D}_j}\ln q_{\boldsymbol{\theta}_D}\left(\boldsymbol{D}_j\right)).
    \end{align}
\end{subequations}
Considering the latent variables $ \{\boldsymbol{H}_j\}_{j=0}^J$ and $\{\boldsymbol{D}_j\}_{j=0}^J$ in the forward diffusion process, the joint \textit{a posteriori} distribution can be rewritten as
\begin{align}
    &p(\left\{\boldsymbol{\varPsi}_j\right\}^J_{j=0}\vert \boldsymbol{Y})=
    p(\boldsymbol{\varPsi}_J\vert \boldsymbol{Y})\prod\nolimits_{j=0}^{J-1} p(\boldsymbol{\varPsi}_j\vert \boldsymbol{\varPsi}_{j+1}, \boldsymbol{Y}),
\end{align}
where $\boldsymbol{\varPsi}_j \triangleq \left\{\boldsymbol{H}_j, \boldsymbol{D}_j\right\}$ for ease of notation. The corresponding variational distribution is
\begin{align}
    &q_{\boldsymbol{\varTheta}}(\left\{\boldsymbol{\varPsi}_j\right\}^J_{j=0}\vert \boldsymbol{Y})=
    q_{\boldsymbol{\varTheta}}\left(\boldsymbol{\varPsi}_J\vert \boldsymbol{Y}\right)\prod\nolimits_{j=0}^{J-1} q_{\boldsymbol{\varTheta}}\left(\boldsymbol{\varPsi}_j\vert \boldsymbol{\varPsi}_{j+1}, \boldsymbol{Y}\right).
\end{align}
Given the forward diffusion process outlined above, we now introduce the following propositions.

\begin{proposition}\label{proposition_transform}
$\mathbb{E}_{p\left(\boldsymbol{Y},\boldsymbol{H}_0,\boldsymbol{D}_0\right)}\left[-\ln q_{\boldsymbol{\varTheta}}\left(\boldsymbol{H}_0,\boldsymbol{D}_0\vert \boldsymbol{Y}\right)\right]$ in \eqref{optimizeP4} is upper-bounded by 
\begin{align}\label{KLdivwithlatent}
    \mathbb{E}_{p(\boldsymbol{Y},\{\boldsymbol{\varPsi}_j\}^J_{j=0})}[-\ln q_{\boldsymbol{\varTheta}}(\{\boldsymbol{\varPsi}_j\}^J_{j=0}\vert \boldsymbol{Y})].
\end{align}
\end{proposition}
\begin{proof}
The cross-entropy \eqref{KLdivwithlatent} can be written as
    \begin{subequations}
    \begin{align}
        &\mathbb{E}_{p(\boldsymbol{Y}, \{\boldsymbol{\varPsi}_j\}^J_{j=0})}[-\ln q_{\boldsymbol{\varTheta}}(\{\boldsymbol{\varPsi}_j\}^J_{j=0}\vert \boldsymbol{Y})]\nonumber\\
        ={}& \mathbb{E}_{p(\boldsymbol{Y},\{\boldsymbol{\varPsi}_j\}^J_{j=0})}[-\ln q_{\boldsymbol{\varTheta}}(\{\boldsymbol{\varPsi}_j\}^J_{j=1}\vert \boldsymbol{\varPsi}_0, \boldsymbol{Y})]\nonumber\\
        &+\mathbb{E}_{p(\boldsymbol{Y},\{\boldsymbol{\varPsi}_j\}^J_{j=0})}[-\ln q_{\boldsymbol{\varTheta}}(\boldsymbol{\varPsi}_0\vert \boldsymbol{Y})]\nonumber\\
        ={}&\mathbb{E}_{p(\boldsymbol{Y},\boldsymbol{\varPsi}_0)}\mathbb{E}_{p(\boldsymbol{Y},\{\boldsymbol{\varPsi}_j\}^J_{j=1}\vert \boldsymbol{\varPsi}_0)}[-\ln q_{\boldsymbol{\varTheta}}(\{\boldsymbol{\varPsi}_j\}^J_{j=1}\vert \boldsymbol{\varPsi}_0, \boldsymbol{Y})]\nonumber\\
        &+
        \mathbb{E}_{p(\boldsymbol{Y},\boldsymbol{\varPsi}_0)}[-\ln q_{\boldsymbol{\varTheta}}(\boldsymbol{\varPsi}_0\vert \boldsymbol{Y})]\label{KLnonneg}\\
        \geq{}& \mathbb{E}_{p(\boldsymbol{Y},\boldsymbol{\varPsi}_0)}[-\ln q_{\boldsymbol{\varTheta}}(\boldsymbol{\varPsi}_0\vert \boldsymbol{Y})],\label{KLnonneg2}
    \end{align}        
    \end{subequations}
    where the inequality in \eqref{KLnonneg2} holds because the first term in \eqref{KLnonneg} is non-negative cross-entropy.
\end{proof}

\begin{proposition}
\label{proposition_KL}
The minimization of the cross-entropy in \eqref{KLdivwithlatent} can be achieved by minimizing
\begin{align}\label{KLdivwithlatent2}
    \mathbb{E}_{p\left(\boldsymbol{\varPsi}_j, \boldsymbol{\varPsi}_{j+1}\vert \boldsymbol{Y}\right)}[-\ln q_{\boldsymbol{\varTheta}}\left(\boldsymbol{\varPsi}_j\vert \boldsymbol{\varPsi}_{j+1}, \boldsymbol{Y}\right)]
\end{align}
from $j=J-1$ to $j=0$.
\end{proposition}

\begin{proof}
The cross-entropy \eqref{KLdivwithlatent} can be factorized as
\begin{align}
    &\mathbb{E}_{p(\{\boldsymbol{\varPsi}_j\}^J_{j=0}\vert \boldsymbol{Y})}[-\ln q_{\boldsymbol{\varTheta}}(\{\boldsymbol{\varPsi}_j\}^J_{j=0}\vert \boldsymbol{Y})] \nonumber\\
    ={}& \sum\nolimits_{j = J-1}^0 \mathbb{E}_{p(\boldsymbol{Y})}\mathbb{E}_{p(\{\boldsymbol{\varPsi}_{j'}\}^J_{j'=0}\vert \boldsymbol{Y})}[-\ln q_{\boldsymbol{\varTheta}}(\boldsymbol{\varPsi}_{j}\vert \boldsymbol{\varPsi}_{j+1}, \boldsymbol{Y})]\nonumber\\
    ={}& \sum_{j = J-1}^0 \mathbb{E}_{p(\boldsymbol{Y})}\mathbb{E}_{p(\boldsymbol{\varPsi}_{j}, \boldsymbol{\varPsi}_{j+1}\vert \boldsymbol{Y})}[-\ln q_{\boldsymbol{\varTheta}}(\boldsymbol{\varPsi}_{j}\vert \boldsymbol{\varPsi}_{j+1}, \boldsymbol{Y})].\label{aux1}
\end{align}
Therefore, we can minimize \eqref{KLdivwithlatent} by separately minimizing the cross-entropies in \eqref{aux1} at each reverse step. 
Since $\mathbb{E}_{p\left(\boldsymbol{\varPsi}_j, \boldsymbol{\varPsi}_{j+1}\vert \boldsymbol{Y}\right)}[-\ln q_{\boldsymbol{\varTheta}}\left(\boldsymbol{\varPsi}_j\vert \boldsymbol{\varPsi}_{j+1}, \boldsymbol{Y}\right)]$ in the $j$-th reverse step
is conditioned on the latent variables in the $(j+1)$-th reverse step, the minimization of \eqref{aux1} is from the $(J-1)$-th reverse step to the $0$-th reverse step.
\end{proof}
Applying propositions \ref{proposition_transform} and \ref{proposition_KL} to the objective in \eqref{optimizeP4}, we transform \eqref{optimizeP4} to equivalent subproblems as
\begin{align}\label{optimizeP5}
    \min_{q_{\boldsymbol{\varTheta}}} \quad & 
    \mathbb{E}_{p\left(\boldsymbol{\varPsi}_j, \boldsymbol{\varPsi}_{j+1}\vert \boldsymbol{Y}\right)}[-\ln q_{\boldsymbol{\varTheta}}\left(\boldsymbol{\varPsi}_j\vert \boldsymbol{\varPsi}_{j+1}, \boldsymbol{Y}\right)]
    ,\nonumber \\ & \quad\quad {j = J-1, \cdots, 0.} 
\end{align}
Solving the subproblems in \eqref{optimizeP5} at each reverse step from $j=J - 1$ to $j = 0$, we obtain 
$q_{\boldsymbol{\varTheta}}\left(\boldsymbol{\varPsi}_0\vert \boldsymbol{\varPsi}_1, \boldsymbol{Y}\right)$ at the $0$-th reverse step, whose mean values are the estimates of the channel and the source data,
i.e.,
\begin{align}
    \{\hat{\boldsymbol{H}}_0,\hat{\boldsymbol{D}}_0\} = g_{\boldsymbol{\varTheta}}(\boldsymbol{Y}) = \mathbb{E}_{q_{\boldsymbol{\varTheta}}\left(\boldsymbol{\varPsi}_0\vert \boldsymbol{\varPsi}_1, \boldsymbol{Y}\right)}[\boldsymbol{\varPsi}_0].
\end{align}




\subsection{Reverse Process of Parallel Variational Diffusion}\label{RPPVD}
In this subsection, we develop the reverse process of PVD to solve the problems in \eqref{optimizeP5}. 
Specifically, the inference model of the blind receiver based on PVD is shown in Fig. \ref{BlindReceiver}, where the dashed arrows represent the reverse process of PVD from step $J-1$ to step $0$.
Different from traditional diffusion models relying on heuristic stochastic sampling \cite{chung2023diffusion}, PVD directly optimizes the variational means of the latent variables in a parallel manner over the reverse steps via gradient descent.

We first model the variational distribution $q_{\boldsymbol{\varTheta}}\left(\boldsymbol{\varPsi}_j\vert \boldsymbol{\varPsi}_{j+1}, \boldsymbol{Y}\right)$ as a Gaussian distribution based on the Gaussianity of $p\left(\boldsymbol{\varPsi}_j\vert \boldsymbol{\varPsi}_{j+1}, \boldsymbol{Y}\right)$ 
given by \cite[Proposition 1]{xue2025score}. Specifically, we have
\begin{align}\label{sample1}
    &q_{\boldsymbol{\varTheta}}\left(\boldsymbol{\varPsi}_j\vert \boldsymbol{\varPsi}_{j+1}, \boldsymbol{Y}\right)
    =\mathcal{CN}(\boldsymbol{H}_j; \hat{\boldsymbol{H}}_j, \Lambda_{H_j}^{-1})
    \mathcal{N}(\boldsymbol{D}_j; \hat{\boldsymbol{D}}_j, \Lambda_{D_j}^{-1}).
\end{align}
where precisions $\Lambda_{H_j}$ and $\Lambda_{D_j}$ can be approximated as
\begin{subequations}\label{Lambdas}
    \begin{align}
        \Lambda_{H_j} &= \frac{\sigma^{2}_{H_{j+1}} + \sigma^{2}_{H_{0}}}{(\sigma^{2}_{H_{j}} + \sigma^{2}_{H_{0}})(\sigma^{2}_{H_{j+1}} - \sigma^{2}_{H_{j}})}, \\
        \Lambda_{D_j} &= \frac{\sigma^{2}_{D_{j+1}} + \sigma^{2}_{D_{0}}}{(\sigma^{2}_{D_{j}} + \sigma^{2}_{D_{0}})(\sigma^{2}_{D_{j+1}} - \sigma^{2}_{D_{j}})}, 
    \end{align}
\end{subequations}
according to \cite[Proposition 1]{xue2025score}. 
We next focus on updating the means $\hat{\boldsymbol{H}}_j$ and $\hat{\boldsymbol{D}}_j$ 
based on $\boldsymbol{H}_{j+1}$, $\boldsymbol{D}_{j+1}$, and $\boldsymbol{Y}$
by solving \eqref{optimizeP5} via gradient descent. Since the variational distribution \eqref{sample1} is Gaussian, following \cite[Eq. (A.2)]{opper2009variational}, the gradients of the objective in \eqref{optimizeP5} to $\hat{\boldsymbol{H}}_j$ and $\hat{\boldsymbol{D}}_j$ are
\begin{subequations}\label{mean}
\begin{align}
    &\nabla_{\hat{\boldsymbol{H}}_j} \mathbb{E}_{p\left(\boldsymbol{\varPsi}_j, \boldsymbol{\varPsi}_{j+1}\vert \boldsymbol{Y}\right)} \left[-\ln q_{\boldsymbol{\varTheta}}\left(\boldsymbol{\varPsi}_j\vert \boldsymbol{\varPsi}_{j+1}, \boldsymbol{Y}\right)\right]
    \nonumber\\ ={}& 
    \mathbb{E}_{p\left(\boldsymbol{\varPsi}_j, \boldsymbol{\varPsi}_{j+1}\vert \boldsymbol{Y}\right)} \left[\nabla_{\boldsymbol{H}_j} \ln q_{\boldsymbol{\varTheta}}\left(\boldsymbol{\varPsi}_j\vert \boldsymbol{\varPsi}_{j+1}, \boldsymbol{Y}\right)\right],\\
    &\nabla_{\hat{\boldsymbol{D}}_j} \mathbb{E}_{p\left(\boldsymbol{\varPsi}_j, \boldsymbol{\varPsi}_{j+1}\vert \boldsymbol{Y}\right)} \left[-\ln q_{\boldsymbol{\varTheta}}(\boldsymbol{\varPsi}_j\vert \boldsymbol{\varPsi}_{j+1}, \boldsymbol{Y})\right]
    \nonumber\\ ={}& 
    \mathbb{E}_{p\left(\boldsymbol{\varPsi}_j, \boldsymbol{\varPsi}_{j+1}\vert \boldsymbol{Y}\right)} \left[\nabla_{\boldsymbol{D}_j} \ln q_{\boldsymbol{\varTheta}}(\boldsymbol{\varPsi}_j\vert \boldsymbol{\varPsi}_{j+1}, \boldsymbol{Y})\right].
\end{align}
\end{subequations}
Thus, the update rules of $\hat{\boldsymbol{H}}_j$ and $\hat{\boldsymbol{D}}_j$ are given by
\begin{subequations}\label{update1}
\begin{align}
    \hat{\boldsymbol{H}}_j &\leftarrow \hat{\boldsymbol{H}}_j - 
    \frac{\epsilon_{H_j}}{L}\sum^{L}_{l=1}\nabla_{\boldsymbol{H}_j} \ln q_{\boldsymbol{\varTheta}}\left(\boldsymbol{\varPsi}_j\vert \boldsymbol{\varPsi}_{j+1}, \boldsymbol{Y}\right)\Big\vert_{\boldsymbol{\varPsi}_j=\boldsymbol{\varPsi}^{(l)}_j},\\
    \hat{\boldsymbol{D}}_j &\leftarrow \hat{\boldsymbol{D}}_j - 
    \frac{\epsilon_{D_j}}{L}\sum^{L}_{l=1}\nabla_{\boldsymbol{D}_j} \ln q_{\boldsymbol{\varTheta}}\left(\boldsymbol{\varPsi}_j\vert \boldsymbol{\varPsi}_{j+1}, \boldsymbol{Y}\right)\Big\vert_{\boldsymbol{\varPsi}_j=\boldsymbol{\varPsi}^{(l)}_j},
\end{align}
\end{subequations}
where $\epsilon_{H_j}$ and $\epsilon_{D_j}$ are the step sizes of gradient descent. The expectations of scores in \eqref{mean} are approximated by the average of $L$ samples in \eqref{update1}. 
Since $\boldsymbol{\varPsi}_{j+1}$ has been given by the $(j+1)$-th reverse step, we have $p\left(\boldsymbol{\varPsi}_j, \boldsymbol{\varPsi}_{j+1}\vert \boldsymbol{Y}\right) \propto p\left(\boldsymbol{\varPsi}_j\vert \boldsymbol{\varPsi}_{j+1}, \boldsymbol{Y}\right)$, and $\{\boldsymbol{\varPsi}^{(l)}_j\}^L_{l = 1}$ can thus be sampled from the variational approximation of $p\left(\boldsymbol{\varPsi}_j\vert \boldsymbol{\varPsi}_{j+1}, \boldsymbol{Y}\right)$ in \eqref{sample1}. 
To leverage $\boldsymbol{H}_{j+1}$, $\boldsymbol{D}_{j+1}$, $\boldsymbol{Y}$ and scores to aid gradient updates, we employ Bayes' rule and obtain $q_{\boldsymbol{\varTheta}}(\boldsymbol{\varPsi}_j\vert \boldsymbol{\varPsi}_{j+1}, \boldsymbol{Y}) = \frac{p(\boldsymbol{\varPsi}_{j+1} \vert\boldsymbol{\varPsi}_j)q_{\boldsymbol{\varTheta}}(\boldsymbol{Y}\vert\boldsymbol{\varPsi}_j)q_{\boldsymbol{\varTheta}}(\boldsymbol{\varPsi}_j)}{q_{\boldsymbol{\varTheta}}(\boldsymbol{\varPsi}_{j+1},\boldsymbol{Y})}$
, where $p(\boldsymbol{\varPsi}_{j+1} \vert\boldsymbol{\varPsi}_j) = p(\boldsymbol{\varPsi}_{j+1} \vert\boldsymbol{\varPsi}_j,\boldsymbol{Y})$. 
The resulting scores are
\begin{subequations}\label{BayesDecomp}
\begin{align}
    \nabla_{\boldsymbol{H}_j} \ln q_{\boldsymbol{\varTheta}}(\boldsymbol{\varPsi}_j\vert \boldsymbol{\varPsi}_{j+1}, \boldsymbol{Y}) &={} 
    \nabla_{\boldsymbol{H}_j} \ln p(\boldsymbol{\varPsi}_{j+1} \vert  \boldsymbol{\varPsi}_j)\nonumber\\
    + \nabla_{\boldsymbol{H}_j} \ln q_{\boldsymbol{\varTheta}}(\boldsymbol{\varPsi}_j) &+ \nabla_{\boldsymbol{H}_j} \ln q_{\boldsymbol{\varTheta}}(\boldsymbol{Y} \vert  \boldsymbol{\varPsi}_j),\\
    \nabla_{\boldsymbol{D}_j} \ln q_{\boldsymbol{\varTheta}}(\boldsymbol{\varPsi}_j\vert \boldsymbol{\varPsi}_{j+1}, \boldsymbol{Y}) &={} 
    \nabla_{\boldsymbol{D}_j} \ln p(\boldsymbol{\varPsi}_{j+1} \vert  \boldsymbol{\varPsi}_j)\nonumber\\
    + \nabla_{\boldsymbol{D}_j} \ln q_{\boldsymbol{\varTheta}}(\boldsymbol{\varPsi}_j) &+ \nabla_{\boldsymbol{D}_j} \ln q_{\boldsymbol{\varTheta}}(\boldsymbol{Y} \vert  \boldsymbol{\varPsi}_j).
\end{align}
\end{subequations}
In \eqref{BayesDecomp}, according to \eqref{conditionaldis}, the first terms on the right-hand sides of the equalities can be calculated explicitly as
\begin{subequations}\label{grad1}
\begin{align}
    \nabla_{\boldsymbol{H}_j} \ln p(\boldsymbol{\varPsi}_{j+1} \vert  \boldsymbol{\varPsi}_j) ={}& \frac{\boldsymbol{H}_{j+1}- \boldsymbol{H}_j}{\sigma^{2}_{H_{j+1}} - \sigma^{2}_{H_{j}}},\\
    \nabla_{\boldsymbol{D}_j} \ln p(\boldsymbol{\varPsi}_{j+1} \vert  \boldsymbol{\varPsi}_j) ={}& \frac{\boldsymbol{D}_{j+1} - \boldsymbol{D}_j}{\sigma^{2}_{D_{j+1}} - \sigma^{2}_{D_{j}}}.
\end{align}
\end{subequations}
The second terms on the right-hand sides of \eqref{BayesDecomp} are the pre-trained first-order score networks $\boldsymbol{S}_{\boldsymbol{\theta}_H}$ and $\boldsymbol{S}_{\boldsymbol{\theta}_D}$ given by \eqref{firstorderscore}.
The third terms on the right-hand sides of \eqref{BayesDecomp} are difficult to calculate due to the intractable variational likelihood function $q_{\boldsymbol{\varTheta}}(\boldsymbol{Y} \vert  \boldsymbol{\varPsi}_j)$. Following the approximations in DPS \cite{chung2023diffusion}, the likelihood scores can be calculated as
\begin{subequations}\label{likelihoodapprox1}
\begin{align}
    \nabla_{\boldsymbol{H}_j} \ln q_{\boldsymbol{\varTheta}}(\boldsymbol{Y} \vert  \boldsymbol{\varPsi}_j)
    ={}&\nabla_{\boldsymbol{H}_j} \ln p(\boldsymbol{Y} \vert  \hat{\boldsymbol{\varPsi}}_{0\vert j}),\\
    \nabla_{\boldsymbol{D}_j} \ln q_{\boldsymbol{\varTheta}}(\boldsymbol{Y} \vert  \boldsymbol{\varPsi}_j)
    ={}&\nabla_{\boldsymbol{D}_j} \ln p(\boldsymbol{Y} \vert  \hat{\boldsymbol{\varPsi}}_{0\vert j}),
\end{align}
\end{subequations}
where $\hat{\boldsymbol{\varPsi}}_{0\vert j} = \{\hat{\boldsymbol{H}}_{0\vert j}, \hat{\boldsymbol{D}}_{0\vert j}\}$, and
\begin{subequations}\label{0givenj}
\begin{align}
    \hat{\boldsymbol{H}}_{0\vert j} &= \mathbb{E}_{q_{\boldsymbol{\varTheta}}(\boldsymbol{H}_0 \vert  \boldsymbol{H}_j)} [\boldsymbol{H}_0] \nonumber\\ &= \boldsymbol{H}_j + \sigma^{2}_{H_j} \nabla_{\boldsymbol{H}_j} \ln q_{\boldsymbol{\varTheta}}(\boldsymbol{H}_j)\nonumber \\ &=\boldsymbol{H}_j + \sigma^{2}_{H_j}\boldsymbol{S}_{\boldsymbol{\theta}_H}(\boldsymbol{H}_j, \sigma_{H_j}),\\
    \hat{\boldsymbol{D}}_{0\vert j} &= \mathbb{E}_{q_{\boldsymbol{\varTheta}}\left(\boldsymbol{D}_0 \vert  \boldsymbol{D}_j\right)} [\boldsymbol{D}_0]  \nonumber\\ &=  \boldsymbol{D}_j + \sigma^{2}_{D_j} \nabla_{\boldsymbol{D}_j} \ln q_{\boldsymbol{\varTheta}}\left(\boldsymbol{D}_j\right) \nonumber\\ &= \boldsymbol{D}_j + \sigma^{2}_{D_j}\boldsymbol{S}_{\boldsymbol{\theta}_D}(\boldsymbol{D}_j, \sigma_{D_j}),
\end{align}
\end{subequations}
are respectively the MMSE estimates of the channel and the source data based on $\boldsymbol{\varPsi}_j$ according to Tweedie's formula \cite{efron2011tweedie}. The likelihood is approximated
as $p(\boldsymbol{Y}\vert \hat{\boldsymbol{\varPsi}}_{0\vert j}) = \mathcal{CN}(\boldsymbol{Y}; \hat{\boldsymbol{H}}_{0\vert j}f_{\boldsymbol{\gamma}}(\hat{\boldsymbol{D}}_{0\vert j}), \sigma^2_n)$. However, the nonlinear coupling between $\hat{\boldsymbol{H}}_{0\vert j}$ and $\hat{\boldsymbol{D}}_{0\vert j}$ caused by the nonlinear encoder $f_{\boldsymbol{\gamma}}$ and the MIMO transmission process will lead to the estimation errors of \eqref{0givenj} to overlap, which results in larger approximation errors of the likelihood scores.
To mitigate such errors, we calibrate the likelihood approximation as follows.
The channel, the source data, and their MMSE estimates satisfy
\begin{subequations}
    \begin{align}
        \boldsymbol{H}_0 ={}& \hat{\boldsymbol{H}}_{0\vert j} + \Delta \boldsymbol{H}_{0\vert j},\\
        \boldsymbol{D}_0 ={}& \hat{\boldsymbol{D}}_{0\vert j} + \Delta \boldsymbol{D}_{0\vert j},
    \end{align}
\end{subequations}
where $\Delta \boldsymbol{H}_{0\vert j}$ (or $\Delta \boldsymbol{D}_{0\vert j}$) is the AWGN independent of $\hat{\boldsymbol{H}}_{0\vert j}$ (or $\hat{\boldsymbol{D}}_{0\vert j}$). The variances of $\Delta \boldsymbol{H}_{0\vert j}$ and $\Delta \boldsymbol{D}_{0\vert j}$ are
\begin{subequations}\label{secondorder}
    \begin{align}
        \sigma^{2}_{H_{0\vert j}} &=\sigma^{2}_{H_j} + \sigma^{4}_{H_j}s_{\boldsymbol{\theta}_{H}}(\boldsymbol{H}_j,\sigma_{H_j})/N_rN_tK^2,\\ 
        \sigma^{2}_{D_{0\vert j}} &=\sigma^{2}_{D_j} + \sigma^{4}_{D_j}s_{\boldsymbol{\theta}_{D}}(\boldsymbol{D}_j,\sigma_{D_j})/n,
    \end{align}
\end{subequations}
according to Tweedie's formula,
where $s_{\boldsymbol{\theta}_{H}}$ and $s_{\boldsymbol{\theta}_{D}}$ are the pre-trained second-order trace scores networks given in \eqref{secondorderscore}, and $n$ is the dimension of the source data.
Recall the system model in \eqref{systemmodel1}-\eqref{encoder}. We thus have the following approximation:
\begin{subequations}
\begin{align}
    \boldsymbol{Y} = {}& \boldsymbol{H}_0 f_{\boldsymbol{\gamma}}(\boldsymbol{D}_0) + \boldsymbol{N} \nonumber \\
    ={}& (\hat{\boldsymbol{H}}_{0\vert j} + \Delta \boldsymbol{H}_{0\vert j})f_{\boldsymbol{\gamma}}(\hat{\boldsymbol{D}}_{0\vert j} + \Delta \boldsymbol{D}_{0\vert j}) + \boldsymbol{N} \nonumber \\
    \approx{}& (\hat{\boldsymbol{H}}_{0\vert j} + \Delta \boldsymbol{H}_{0\vert j})(f_{\boldsymbol{\gamma}}(\hat{\boldsymbol{D}}_{0\vert j}) \nonumber\\&+ \mathcal{J}_{f_{\boldsymbol{\gamma}}}(\hat{\boldsymbol{D}}_{0\vert j})\Delta \boldsymbol{D}_{0\vert j}) + \boldsymbol{N} \label{Taylor} \\
    ={}& \hat{\boldsymbol{H}}_{0\vert j}f_{\boldsymbol{\gamma}}(\hat{\boldsymbol{D}}_{0\vert j}) + \boldsymbol{N} \nonumber \\ &+ \Delta\boldsymbol{H}_{0\vert j}f_{\boldsymbol{\gamma}}(\hat{\boldsymbol{D}}_{0\vert j}) + \hat{\boldsymbol{H}}_{0\vert j}\mathcal{J}_{f_{\boldsymbol{\gamma}}}(\hat{\boldsymbol{D}}_{0\vert j})\Delta \boldsymbol{D}_{0\vert j} \nonumber \\ &+ \Delta \boldsymbol{H}_{0\vert j}\mathcal{J}_{f_{\boldsymbol{\gamma}}}(\hat{\boldsymbol{D}}_{0\vert j})\Delta \boldsymbol{D}_{0\vert j} \label{threeterms}\\
    ={}& \hat{\boldsymbol{H}}_{0\vert j}f_{\boldsymbol{\gamma}}(\hat{\boldsymbol{D}}_{0\vert j}) + \Delta \boldsymbol{N}_{0\vert j} + \boldsymbol{N},
\end{align}
\end{subequations}
where we apply first-order Taylor expansion to $f_{\boldsymbol{\gamma}}\left(\boldsymbol{D}_0 \right)$ at $\hat{\boldsymbol{D}}_{0\vert j}$ in \eqref{Taylor}, with $\mathcal{J}_{f_{\boldsymbol{\gamma}}}(\hat{\boldsymbol{D}}_{0\vert j})$ being the Jacobian matrix. The last three terms in \eqref{threeterms} are aggregated into the noise $\Delta \boldsymbol{N}_{0\vert j}$, which is approximated as AWGN with mean
zero and variance
\begin{align}
    \sigma_{\Delta N_{0\vert j}}^2 = 
    \frac{\left\Vert\Delta\boldsymbol{N}_{0\vert j}\right\Vert_F^2}{N_rKT}.
\end{align}
Therefore, the likelihood scores in \eqref{likelihoodapprox1} are calculated as
\begin{subequations}\label{grad3}
\begin{align}
    \hspace{-0.1cm}\nabla_{\boldsymbol{H}_j} \ln q_{\boldsymbol{\varTheta}}(\boldsymbol{Y} \vert  \boldsymbol{\varPsi}_j)={}&\nabla_{\boldsymbol{H}_j} \frac{-\Vert\boldsymbol{Y} - \hat{\boldsymbol{H}}_{0\vert j}f_{\boldsymbol{\gamma}}(\hat{\boldsymbol{D}}_{0\vert j})\Vert_F^2}{\sigma_{\Delta N_{0\vert j}}^2+\sigma_n^2},\\
    \hspace{-0.1cm}\nabla_{\boldsymbol{D}_j} \ln q_{\boldsymbol{\varTheta}}(\boldsymbol{Y} \vert  \boldsymbol{\varPsi}_j)={}&\nabla_{\boldsymbol{D}_j} \frac{-\Vert\boldsymbol{Y} - \hat{\boldsymbol{H}}_{0\vert j}f_{\boldsymbol{\gamma}}(\hat{\boldsymbol{D}}_{0\vert j})\Vert_F^2}{\sigma_{\Delta N_{0\vert j}}^2+\sigma_n^2},
\end{align}
\end{subequations}
where the estimation errors of $\hat{\boldsymbol{H}}_{0\vert j}$ and $\hat{\boldsymbol{D}}_{0\vert j}$ are taken into consideration in the likelihood approximation. 
Finally, the scores in \eqref{update1} is obtained by using \eqref{BayesDecomp}, \eqref{firstorderscore}, \eqref{grad1}, and \eqref{grad3}.

\subsection{Overall Algorithm 
}
We summarize the details of the PVD algorithm for the blind receiver design of the Blind-MIMOSC framework in Algorithm \ref{alg1}. 
PVD simultaneously recovers the channel and the source data from input signal $\boldsymbol{Y}$ by leveraging the known DJSCC encoder $f_{\boldsymbol{\gamma}}$ and pre-trained score networks parameterized by $\boldsymbol{\varTheta}$. During the $J$ outer loops, PVD executes the reverse diffusion process. Within the $j$-th reverse diffusion step, PVD performs gradient descent during $J_{\text{in}}$ inner loops to optimize the means of latent variables $\boldsymbol{H}_j$ and $\boldsymbol{D}_j$, which will be used in the next reverse step. After $J$ reverse steps, PVD outputs variational means $\hat{\boldsymbol{H}}_0$ and $\hat{\boldsymbol{D}}_0$ as the recovered channel and source data, respectively. 


Unlike end-to-end schemes such as DJSCC-MIMO \cite{wu2024deep}, PVD leverages pre-trained, unconditional score-based models in a plug-and-play fashion, preventing costly joint retraining whenever source datasets or channel conditions vary.  
Traditional score-based approaches like DPS \cite{chung2023diffusion} rely on stochastic sample generation to recover the source data, but random posterior sampling is likely to fail under the complicated nonlinear coding and channel mixing conditions. In contrast, PVD minimizes the cross-entropy by optimizing the variational means via deterministic gradient descent over the reverse diffusion steps. This deterministic optimization delivers more stable and faster reconstructions compared to DPS. 


The parameter settings in Algorithm \ref{alg1} are specified as follows. 
By fixing $\sigma_{H_J}$, $\sigma_{D_J}$, $\sigma_{H_1}$, and $\sigma_{D_1}$, the pre-defined noise variances follow an exponential interpolation as $\sigma_{H_{j}} = \sigma_{H_{1}}(\frac{\sigma_{H_J}}{\sigma_{H_{1}}})^{\frac{j}{J}}$ and $\sigma_{D_{j}} = \sigma_{D_{1}}(\frac{\sigma_{D_J}}{\sigma_{D_{1}}})^{\frac{j}{J}}$. The step sizes are set as $\epsilon_{H_j} = \zeta_{H}(\sigma^2_{H_{j+1}}-\sigma^2_{H_j})$ and $\epsilon_{D_j} = \zeta_{D}(\sigma^2_{D_{j+1}}-\sigma^2_{D_j})$, where $\zeta_{H}$ and $\zeta_{D}$ are fixed hyper-parameters. This formulation ensures that step sizes diminish progressively during the reverse diffusion process. For initialization, $\boldsymbol{H}_J$ and $\hat{\boldsymbol{H}}_{J-1}$ are drawn from zero-mean complex Gaussian distributions with variances $\sigma^2_{H_J}$ and $\sigma^2_{H_{J-1}}$, respectively. Similarly, $\boldsymbol{D}_J$ and $\hat{\boldsymbol{D}}_{J-1}$ are drawn from zero-mean Gaussian distributions with variances $\sigma^2_{D_J}$ and $\sigma^2_{D_{J-1}}$, respectively.

\begin{algorithm}[htb]
    \caption{PVD Algorithm}\label{alg1}
    \KwIn{$\boldsymbol{\gamma}$, $\boldsymbol{\varTheta}$, $\boldsymbol{Y}$, $\sigma^2_n$, $J$, $J_{\text{in}}$, $\{\sigma_{H_j},\sigma_{D_j}\}^{J}_{j=0}$, $L$, $\{\epsilon_{H_j},\epsilon_{D_j}\}^{J-1}_{j=0}$;}
    \textbf{Initialization:} $\boldsymbol{H}_J$, $\boldsymbol{D}_J$, $\hat{\boldsymbol{H}}_{J-1}$, $\hat{\boldsymbol{D}}_{J-1}$;\\
    \For{$j = J-1 : 0$}
        {
        Update $\Lambda_{H_j}$ and $\Lambda_{D_j}$ by \eqref{Lambdas}; \\
            \For{$j_{\text{in}}=1:J_{\text{in}}$}
            {
            Sample $\boldsymbol{H}_j$ and $\boldsymbol{D}_j$ from \eqref{sample1};\\
            Update $\hat{\boldsymbol{H}}_{0\vert j}$ and $\hat{\boldsymbol{D}}_{0\vert j}$ by using \eqref{0givenj}; \\
            Update $\hat{\boldsymbol{H}}_j$ and $\hat{\boldsymbol{D}}_j$ by using \eqref{update1};\\
            }
        $\boldsymbol{H}_j = \hat{\boldsymbol{H}}_j$, $\boldsymbol{D}_j = \hat{\boldsymbol{D}}_j$,\\ $\hat{\boldsymbol{H}}_{j-1} = \hat{\boldsymbol{H}}_j$, $\hat{\boldsymbol{D}}_{j-1} = \hat{\boldsymbol{D}}_j$
        }
    \KwOut{$\hat{\boldsymbol{H}}_0$, $\hat{\boldsymbol{D}}_0$.}
\end{algorithm}

\section{Extension to Multi-User MIMO}
In this section, we consider a MIMO system with multiple users. First, we modify the system model \eqref{systemmodel1} by superposing signals from multiple users. Then, parallel variational diffusion is extended to multi-user systems to recover multiple channels and source data simultaneously.

\subsection{Multi-user System Model}
Consider a multi-user MIMO uplink system where the receiver, equipped with $N_r$ antennas, communicates to $N_u$ transmitters, each equipped with $N_t$ antennas. The compound MIMO channel between the $i$-th user and the receiver is denoted by $\boldsymbol{H}^{(i)}_0 \in \mathbb{C}^{N_rK \times N_tK}$. The source data $\boldsymbol{D}^{(i)}_{0}$ of the $i$-th user is mapped into the transmitted signal matrix 
\begin{align}\label{encoder2}
    \boldsymbol{X}^{(i)} = f_{\boldsymbol{\gamma}_i}(\boldsymbol{D}^{(i)}_0)\in\mathbb{C}^{N_tK\times T},
\end{align} 
through a DJSCC encoder $f_{\boldsymbol{\gamma}_i}$ parameterized by $\boldsymbol{\gamma}_i$.
At the receiver, a superposition of signals from all $N_u$ users across all blocks is represented by
\begin{align}\label{systemmodel2}
    \boldsymbol{Y}
    = \sum\nolimits_{i = 1}^{N_u}\boldsymbol{H}^{(i)}_{0}\boldsymbol{X}^{(i)} + \boldsymbol{N},
\end{align}
where $\boldsymbol{N} \in \mathbb{C}^{N_rK \times T}$ is AWGN. The diffusion-based blind receiver $g_{\boldsymbol{\varTheta}}$ parameterized by $\boldsymbol{\varTheta}$ jointly estimates the channels $\{\boldsymbol{H}^{(i)}_0\}^{N_u}_{i=1}$ and the source data $\{\boldsymbol{D}^{(i)}_0\}^{N_u}_{i=1}$ directly from the received signal $\boldsymbol{Y}$.

\subsection{Extended Parallel Variational Diffusion}
Based on the channel model in \eqref{systemmodel2} and the blind receiver design in Section~\eqref{Section3}, the diffusion-based blind receiver in the multi-user MIMO semantic communication system can be achieved by minimizing the cross-entropy as
\begin{align}\label{optimizeP6}
    \min_{q_{\boldsymbol{\varTheta}}} \,\,
    \mathbb{E}_{p(\boldsymbol{Y},\{\boldsymbol{H}^{(i)}_0,\boldsymbol{D}^{(i)}_0\}^{N_u}_{i=1})}[-\ln q_{\boldsymbol{\varTheta}}(\{\boldsymbol{H}^{(i)}_0,\boldsymbol{D}^{(i)}_0\}^{N_u}_{i=1}\vert\boldsymbol{Y})].
\end{align}
To exploit score-based models, we introduce the latent variables of the channels and the source data as $\{\boldsymbol{H}^{(i)}_j, \boldsymbol{D}^{(i)}_j\}^{N_u,J}_{i=1,j=1}$. The pre-defined noise variances for the latent variables are $\sigma^{2}_{H^{(i)}_{J}} > \cdots> \sigma^{2}_{H^{(i)}_{1}} > \sigma^{2}_{H^{(i)}_{0}} = 0$ and $\sigma^{2}_{D^{(i)}_{J}} > \cdots> \sigma^{2}_{D^{(i)}_{1}} > \sigma^{2}_{D^{(i)}_{0}} = 0$ for any $i$. 
By redefining $\boldsymbol{\varPsi}_j \triangleq \{\boldsymbol{H}^{(i)}_j, \boldsymbol{D}^{(i)}_j\}^{N_u}_{i=1}$, the derivations for the reverse process in Section~\ref{RPPVD} literarily hold for the multi-user case. Specifically, applying propositions \ref{proposition_transform} and \ref{proposition_KL} to \eqref{optimizeP6}, the optimization problem is transformed to
\begin{align}\label{optimizeP7}
    \min_{q_{\boldsymbol{\varTheta}}} \quad & \mathbb{E}_{p\left(\boldsymbol{\varPsi}_j, \boldsymbol{\varPsi}_{j+1}\vert \boldsymbol{Y}\right)}[-\ln q_{\boldsymbol{\varTheta}}\left(\boldsymbol{\varPsi}_j\vert \boldsymbol{\varPsi}_{j+1}, \boldsymbol{Y}\right)]
    ,\nonumber \\ & \quad\quad {j = J-1, \cdots, 0.} 
\end{align}
At each reverse step, we optimize the mean values $\{\hat{\boldsymbol{H}}^{(i)}_j,\hat{\boldsymbol{D}}^{(i)}_j\}^{N_u}_{i=1}$ of the variational Gaussian distribution 
\begin{align}\label{sample1???}
    &q_{\boldsymbol{\varTheta}}(\boldsymbol{\varPsi}_j\vert \boldsymbol{\varPsi}_{j+1}, \boldsymbol{Y})\nonumber\\
    ={}& \prod^{N_u}_{i=1}\mathcal{CN}(\boldsymbol{H}^{(i)}_j; \hat{\boldsymbol{H}}^{(i)}_j, \Lambda_{H^{(i)}_j}^{-1})
    \mathcal{N}(\boldsymbol{D}^{(i)}_j; \hat{\boldsymbol{D}}^{(i)}_j, \Lambda_{D^{(i)}_j}^{-1}),
\end{align}
to solve \eqref{optimizeP7}, where precisions $\{\Lambda_{H^{(i)}_j}, \Lambda_{D^{(i)}_j}\}^{N_u,J-1}_{i=1,j=0}$ can be approximately obtained similarly to \eqref{Lambdas}.
The gradients of the objective in \eqref{optimizeP7} to the variational means are given by
\begin{subequations}\label{mean???}
\begin{align}
    &\nabla_{\hat{\boldsymbol{H}}^{(i)}_j} \mathbb{E}_{p(\boldsymbol{\varPsi}_j, \boldsymbol{\varPsi}_{j+1}\vert \boldsymbol{Y})} \left[-\ln q_{\boldsymbol{\varTheta}}\left(\boldsymbol{\varPsi}_j\vert \boldsymbol{\varPsi}_{j+1}, \boldsymbol{Y}\right)\right] 
    \nonumber\\ ={}& 
    \mathbb{E}_{p(\boldsymbol{\varPsi}_j, \boldsymbol{\varPsi}_{j+1}\vert \boldsymbol{Y})} [\nabla_{\boldsymbol{H}^{(i)}_j} \ln q_{\boldsymbol{\varTheta}}\left(\boldsymbol{\varPsi}_j\vert \boldsymbol{\varPsi}_{j+1}, \boldsymbol{Y}\right)],\\
    &\nabla_{\hat{\boldsymbol{D}}^{(i)}_j} \mathbb{E}_{p(\boldsymbol{\varPsi}_j, \boldsymbol{\varPsi}_{j+1}\vert \boldsymbol{Y})} [-\ln q_{\boldsymbol{\varTheta}}(\boldsymbol{\varPsi}_j\vert \boldsymbol{\varPsi}_{j+1}, \boldsymbol{Y})] 
    \nonumber\\ ={}& 
    \mathbb{E}_{p(\boldsymbol{\varPsi}_j, \boldsymbol{\varPsi}_{j+1}\vert \boldsymbol{Y})} [\nabla_{\boldsymbol{D}^{(i)}_j} \ln q_{\boldsymbol{\varTheta}}(\boldsymbol{\varPsi}_j\vert \boldsymbol{\varPsi}_{j+1}, \boldsymbol{Y})],
\end{align}
\end{subequations}
for $i\in\{1,\cdots,N_u\}$, which can be calculated by employing Bayes' rule as in \eqref{BayesDecomp}. 
The update rules of the variational means are thus given by
\begin{subequations}\label{update???}
\begin{align}
    \hat{\boldsymbol{H}}^{(i)}_j &\leftarrow \hat{\boldsymbol{H}}^{(i)}_j - 
    \frac{\epsilon_{H^{(i)}_j}}{L}\sum^{L}_{l=1}\nabla_{\boldsymbol{H}^{(i)}_j} \ln q_{\boldsymbol{\varTheta}}\left(\boldsymbol{\varPsi}_j\vert \boldsymbol{\varPsi}_{j+1}, \boldsymbol{Y}\right)\Big\vert_{\boldsymbol{\varPsi}_j=\boldsymbol{\varPsi}^{(l)}_j},\\
    \hat{\boldsymbol{D}}^{(i)}_j &\leftarrow \hat{\boldsymbol{D}}^{(i)}_j - 
    \frac{\epsilon_{D^{(i)}_j}}{L}\sum^{L}_{l=1}\nabla_{\boldsymbol{D}^{(i)}_j} \ln q_{\boldsymbol{\varTheta}}\left(\boldsymbol{\varPsi}_j\vert \boldsymbol{\varPsi}_{j+1}, \boldsymbol{Y}\right)\Big\vert_{\boldsymbol{\varPsi}_j=\boldsymbol{\varPsi}^{(l)}_j},
\end{align}
\end{subequations}
where $\{\epsilon_{H^{(i)}_j},\epsilon_{D^{(i)}_j}\}^{N_u,J-1}_{i=1,j=0}$ are the step sizes, and $\{\boldsymbol{\varPsi}^{(l)}_j\}^L_{l = 1}$ can be sampled from the variational distribution in \eqref{sample1???}. Following the proposed reverse diffusion process of the PVD model, the channels and the source data can be simultaneously recovered via \eqref{update???} from $j=J-1$ to $j=0$.

\section{Numerical Results}\label{NR}
In this section, we conduct numerical experiments to evaluate the performance of our proposed Blind-MIMOSC framework with the PVD model over MIMO channels in various fading and SNR scenarios. Specifically, we consider the image source data for wireless transmission and conduct experiments on the FFHQ-256 validation dataset \cite{karras2019style}. The DJSCC encoders
are pre-trained based on the Swin Transformer backbone \cite{yang2024swinjscc}. The score-based models for the channel and image source are pre-trained based on the improved noise conditional score networks (NCSN++) \cite{song2021scorebased}. The parameters in Algorithm \ref{alg1} are $J=30$, $J_{\text{in}}=20$, $L=1$, $\sigma_{H^{(i)}_{J}} = \sigma_{D^{(i)}_{J}} = 100$ and $\sigma_{H^{(i)}_{1}} = \sigma_{D^{(i)}_{1}} = 0.01$ for $i\in\{1,\cdots,N_u\}$. The step sizes are tuned independently across different experimental setups to optimize performance. The experiments are conducted on a server equipped with 10 NVIDIA GeForce RTX 3090 GPUs.

\subsection{Baselines}
To comprehensively evaluate the proposed method, we perform comparative analyses against the existing transmission strategies listed below:

\begin{enumerate}
    \item \textbf{BPG-LDPC + pilots:} 
    This scheme uses the BPG codec 
    for source compression, 5G LDPC codes with a block length of 4096 for channel coding, and modulation formats compliant with 3GPP TS 38.214.
    Pilot symbols are used for channel estimation.
    An exhaustive search over all feasible combinations of coding rates and modulation constellation sizes is performed to identify the optimal configuration, serving as the benchmark.

    \item \textbf{DPS-MIMO + pilots:} 
    This system utilizes our proposed DJSCC encoder $f_{\boldsymbol{\gamma}}$ at the transmitter. At the receiver, channel estimation is performed using pilots, and source reconstruction is achieved through DPS \cite{chung2023diffusion} based on the estimated channel and the encoder. 
    
    \item \textbf{DJSCC-MIMO + pilots:} The DJSCC-MIMO codec \cite{wu2024deep} with pilot-based channel estimation are employed.
    
    \item \textbf{DJSCC + Pro-BiG-AMP:} This scheme integrates our proposed DJSCC encoder $f_{\boldsymbol{\gamma}}$ and the projection-based BiG-AMP (Pro-BiG-AMP) algorithm for blind channel estimation and signal detection \cite{zhang2018blind}. $\lceil 1 + \log N_u\rceil N_t$ reference symbols are inserted into the transmitted signal in each transmission block to remove ambiguity.
    
    \item \textbf{DJSCC-MIMO + Perfect CSI:} As an ideal reference, the DJSCC-MIMO codec \cite{wu2024deep} with perfect CSI known at the receiver is included.
    
    \item \textbf{BPG-Capacity + Perfect CSI:} As an ideal reference, this framework assumes capacity-achieving channel codes and uses BPG as the source codec. 
    Transmitters are assumed to know the instantaneous capacity of each channel block and select the optimal compression rate of BPG based on it.

\end{enumerate}



\begin{figure*}[htb]
    \centering
    \setlength{\tabcolsep}{0.45pt} 
    \renewcommand{\arraystretch}{0} 
    \begin{tabular}{ccccccc}
        & 
        \multicolumn{1}{c}{\scriptsize \textbf{Blind-MIMOSC + }} & 
        \multicolumn{1}{c}{\scriptsize BPG-Capacity + } & 
        \multicolumn{1}{c}{\scriptsize BPG-LDPC + } & 
        \multicolumn{1}{c}{\scriptsize DJSCC-MIMO + } & 
        \multicolumn{1}{c}{\scriptsize DJSCC-MIMO + } & 
        \multicolumn{1}{c}{\scriptsize DPS-MIMO + } \\[2pt]

        \multicolumn{1}{c}{\scriptsize Ground truth} & 
        \multicolumn{1}{c}{\scriptsize \textbf{PVD}} & 
        \multicolumn{1}{c}{\scriptsize Perfect CSI}& 
        \multicolumn{1}{c}{\scriptsize $2N_t$ pilots} & 
        \multicolumn{1}{c}{\scriptsize Perfect CSI} & 
        \multicolumn{1}{c}{\scriptsize $2N_t$ pilots} & 
        \multicolumn{1}{c}{\scriptsize $2N_t$ pilots} \\[2pt]
        
        \multicolumn{1}{c}{} & 
        \multicolumn{1}{c}{\scriptsize $\text{CBR}=0.023$} & 
        \multicolumn{1}{c}{\scriptsize $\text{CBR}=0.023$} & 
        \multicolumn{1}{c}{\scriptsize $\text{CBR}=0.035$ {\color{red}$\uparrow 51.7\%$}} & 
        \multicolumn{1}{c}{\scriptsize $\text{CBR}=0.023$} & 
        \multicolumn{1}{c}{\scriptsize $\text{CBR}=0.035$ {\color{red}$\uparrow 52.2\%$}} &
        \multicolumn{1}{c}{\scriptsize $\text{CBR}=0.035$ {\color{red}$\uparrow 52.2\%$}} \\[2pt]
        
        \includegraphics[width=0.141\linewidth]{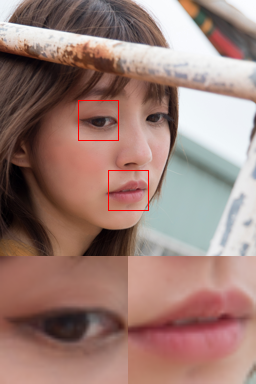} &
        \includegraphics[width=0.141\linewidth]{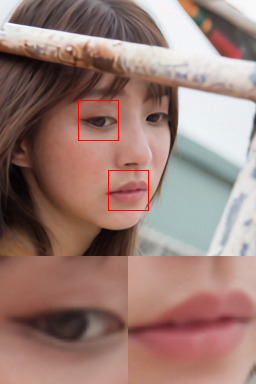} &
        \includegraphics[width=0.141\linewidth]{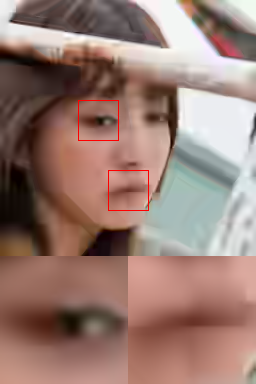} &
        \includegraphics[width=0.141\linewidth]{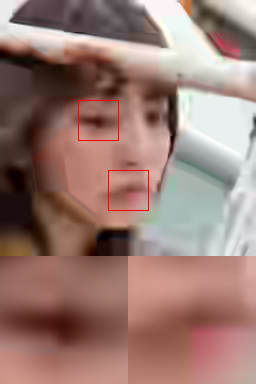}  &
        \includegraphics[width=0.141\linewidth]{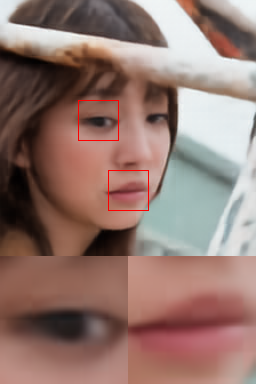} &
        \includegraphics[width=0.141\linewidth]{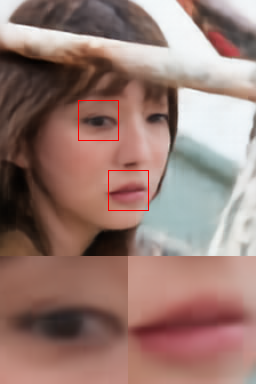} &
        \includegraphics[width=0.141\linewidth]{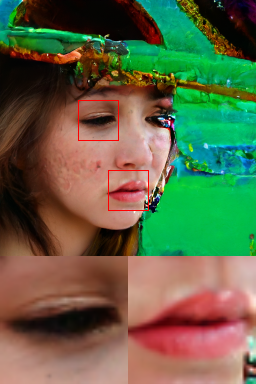}
    \end{tabular}
    \caption{Visual comparisons of image source under FFHQ-256 validation set and Rayleigh fading channel dataset. Parameters are set as $N_r=4$, $N_u = 4$, $N_t = 1$, SNR $=20$ dB. The channel remains constant for $T=24$ channel uses. 
    The original images are displayed in the first column. The remaining columns show the images recovered by Blind-MIMOSC + PVD, BPG-Capacity + Perfect CSI, BPG-LDPC + $2N_t$ pilots, DJSCC-MIMO + perfect CSI, DJSCC-MIMO + $2N_t$ pilots, and DPS-MIMO + $2N_t$ pilots, respectively. 
    Note that DJSCC + Pro-BiG-AMP always fails to recover the images in this setting, so we omit it to save space.
    }
    \vspace{-0.3cm}
    \label{imagecomp}
\end{figure*}

\subsection{Performance Metrics}
Given the channel output $\boldsymbol{Y}$, the receiver reconstructs the source images as $\hat{\boldsymbol{D}}^{(i)}_0\in\mathbb{R}^{256\times 256\times3}$ and the channels as $\hat{\boldsymbol{H}}^{(i)}_0\in \mathbb{C}^{N_rK \times N_tK}$ for any $i=1,\cdots,N_u$. The reconstruction quality of images is quantified through three metrics:
1) \textbf{MS-SSIM} \cite{wang2003multiscale} measures the quality of images by comparing luminance, contrast, and structural information between the original and reconstructed images. A higher MS-SSIM score indicates a better structural similarity; 2) \textbf{DISTS} \cite{ding2020image} is designed to capture both the structural and texture differences between the original and reconstructed images using deep features. A lower DISTS score indicates a higher similarity in both structure and perception; 3) \textbf{LPIPS} \cite{Zhang_2018_CVPR} computes the dissimilarity between the feature vectors of the original and reconstructed images. A lower LPIPS score implies a better perceptual quality.
The channel estimation performance is measured via NMSE, defined as
\begin{align}
    \text{NMSE} = 10\log_{10} \left(\sum^{N_u}_{i=1}\frac{\Vert  \boldsymbol{H}^{(i)}_0-\hat{\boldsymbol{H}}^{(i)}_{0}\Vert^2_F}{N_u\Vert  \boldsymbol{H}^{(i)}_0\Vert^2_F}\right) .
\end{align} 
A lower NMSE implies a better estimation accuracy of the channel matrix.
The channel bandwidth ratio is defined as $\text{CBR} \triangleq \frac{N_tKT}{n}$ \cite{wu2024deep}, representing the number of signal matrix elements per source data dimension, where $n=256\times 256\times 3$. 
We use black ``$\uparrow$'' and ``$\downarrow$'' to mark that lower or higher values of the corresponding metrics represent better quality. 
To measure channel quality, we define the channel SNR as
\begin{align}
    \text{SNR} = 10\log_{10} \frac{\Vert  \sum^{N_u}_{i=1}\boldsymbol{H}^{(i)}_0\boldsymbol{X}^{(i)}\Vert^2_F}{\Vert  \boldsymbol{N}\Vert^2_F} .
\end{align} 
All the results are conducted by averaging over at least 300 Monte Carlo trials.

\subsection{Rayleigh Fading Channel Model}
We first adopt the Rayleigh fading channel model for channel generation, i.e., the channel matrix in our experiment is generated using i.i.d. complex Gaussian random variables. To ensure a fair comparison, we employ the linear minimum mean-squared error (LMMSE) channel estimator for pilot-based approaches, which is optimal in the NMSE sense under Rayleigh fading channels.

We compare the performance of our proposed algorithm with several baselines. The parameters are set as follows. The antenna array size of the receiver $N_r = 4$, the number of users $N_u = 4$, the antenna array size of each user $N_t = 1$, and SNR $=20$ dB. Each image is encoded into $4608$ symbols for transmission, and the channel remains constant for $T = 24$ channel uses\footnote{This configuration references typical parameters from the 3GPP 5G NR standard: with a subcarrier spacing of 30 kHz,
the total duration of 24 symbols corresponds to 0.8 ms, which serves as a reasonable setting of the practical channel coherence time.}. For our proposed Blind-MIMOSC framework, the number of transmission blocks is $K = \frac{4608}{N_tT}$. 
For the pilot-based algorithms, the transmitted signal matrix is $\boldsymbol{X} = [\boldsymbol{X}_p, f_{\boldsymbol{\gamma}}(\boldsymbol{D}_0)]$, where $\boldsymbol{X}_p\in\mathbb{C}^{N_tK\times N_p}$ is the pilot matrix with $N_p$ being the number of pilots transmitted by each antenna in a transmission block. The dimension of the encoder output matrix will be $N_tK \times (T-N_p)$, accordingly. Thus, the number of transmission blocks becomes $K = \frac{4608}{N_t(T-N_p)}$.

Fig. \ref{imagecomp} illustrates the recovered images for visual comparisons.
The first column displays the original images,
while the remaining columns show the images recovered by Blind-MIMOSC + PVD, BPG-Capacity + Perfect CSI, BPG-LDPC + $2N_t$ pilots, DJSCC-MIMO + perfect CSI, DJSCC-MIMO + $2N_t$ pilots, and DPS-MIMO + $2N_t$ pilots, respectively. Note that Pro-BiG-AMP \cite{zhang2018blind} fails to recover the images in this setup, and thus is omitted for brevity. Fig. \ref{imagecomp} shows that our proposed Blind-MIMOSC with PVD delivers visually more pleasing results, such as the details around the eyes and mouth, at a much lower CBR.
\begin{figure}[htb]
    \centering
        \includegraphics[width=0.99\linewidth]{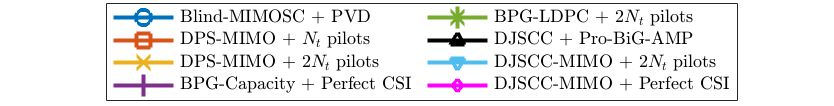}
        \includegraphics[width=0.49\linewidth]{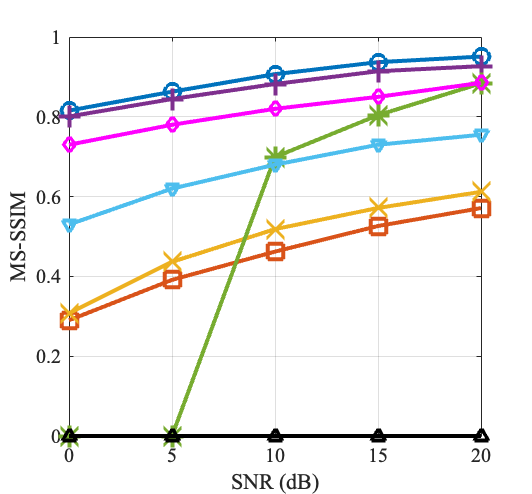}
        \includegraphics[width=0.49\linewidth]{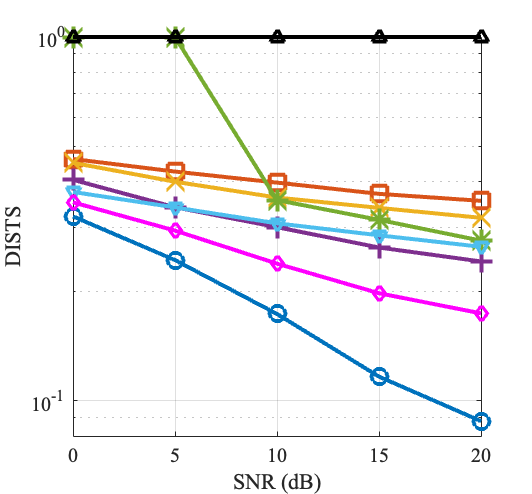}\\
        \includegraphics[width=0.49\linewidth]{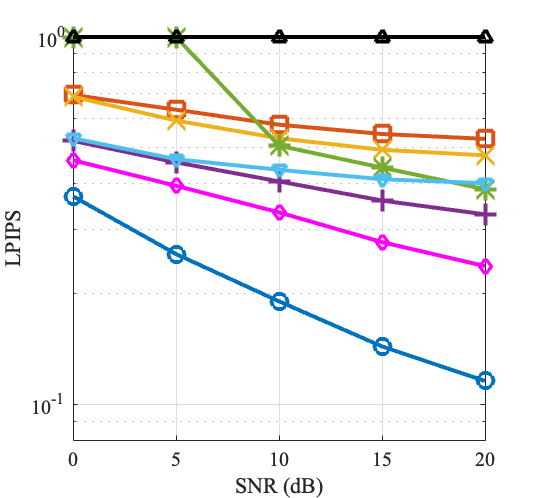}
        \includegraphics[width=0.49\linewidth]{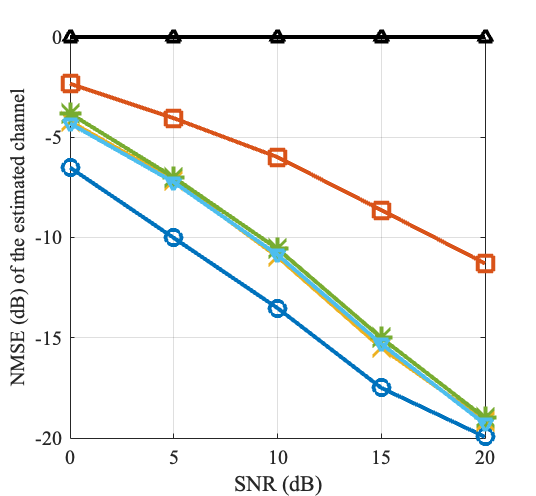}
    \vspace{-0.3cm}
    \caption{MS-SSIM, DISTS, LPIPS of recovered images and NMSE (dB) of channel estimation vs. different SNR levels over the FFHQ-256 validation set. Parameters are set as $N_r = 4$, $N_u = 4$, and $N_t = 1$.}
    \vspace{-0.15cm}
    \label{PLN}
\end{figure}
The quantitative results of the reconstruction quality of the source, such as MS-SSIM, DISTS, LPIPS, and channel estimation NMSE versus different SNRs are shown in Fig. \ref{PLN}. These results indicate that our Blind-MIMOSC with PVD dominates among all practical baselines with lower CBR, achieving at least $7.6\%$ higher MS-SSIM, $66.8\%$ lower DISTS, and $69.9\%$ lower LPIPS compared to the second-best practical approach.
Compared to the ideal reference, BPG-Capacity + Perfect CSI, Blind-MIMOSC with PVD preserves a competitive performance in terms of MS-SSIM, and a much better performance in DISTS and LPIPS metrics, which implies that PVD can better recover the details of the source data.
As for channel estimation, Blind-MIMOSC drastically reduces NMSE by up to 8.9 dB compared to existing schemes.
This gain arises from PVD’s ability to exploit the structure in the source data to aid channel estimation.

We also test the proposed Blind-MIMOSC in the massive MIMO case, where $N_r = 64$, $N_t = 8$, $N_u = 1$, and SNR $=0$ dB. Each image is encoded into $4608$ symbols for transmission, and the channel remains constant for $T=24$ channel uses. The resulting quantitative results are shown in the Table. \ref{massiveMIMO}. It can be seen that in the massive MIMO scenario, the advantages of Blind-MIMOSC are further highlighted. In terms of source recovery performance, Blind-MIMOSC continues to maintain a significant advantage over the practical pilot-based approaches, with at least $67.1\%$ higher MS-SSIM, $67.0\%$ lower DISTS, and $65.1\%$ lower LPIPS. The NMSE of channel estimates of the Blind-MIMOSC is reduced by 2.7 dB and 0.9 dB compared to pilot‑based approaches employing $N_p = N_t$ and $N_p = 2N_t$ pilots, respectively. While in terms of transmission efficiency, since the pilots in pilot-based approaches increase linearly with the number of transmit antennas, the CBR of Blind-MIMOSC can be only \textbf{1/3} of the CBR of pilot-based approaches.
\begin{table}[htb]
\centering
\setlength{\tabcolsep}{0.7pt}
\caption{Quantitative results on FFHQ-256 dataset and Rayleigh fading channels. $N_r = 64$, $N_t = 8$, $N_u = 1$, $T = 24$, SNR $=0$ dB. \textbf{Bold}: best \& \underline{Underline}: second-best among practical approaches.}
\label{massiveMIMO}
\begin{tabular}{l | c | c c c |c }
    \toprule
    \multirow{3}{*}{Methods} &  \multirow{3}{*}{CBR}&\multirow{3}{*}{MS-SSIM$\uparrow$} & \multirow{3}{*}{DISTS$\downarrow$} & \multirow{3}{*}{LPIPS$\downarrow$} & Channel \\
    &  && & & Estimation \\
    &  && & & NMSE$\downarrow$ \\
    \midrule
    BPG-Capacity + Perfect CSI &  0.0234&0.9653& 0.1718& 0.2208 & N/A \\
    DJSCC-MIMO + Perfect CSI &      0.0234  &0.7022& 0.2796& 0.3138 & N/A \\
    \midrule
    Blind-MIMOSC + PVD &  \textbf{0.0234}  &\textbf{0.9085}& \textbf{0.1398}& \textbf{0.2070} & \textbf{-4.7221} \\
    DPS-MIMO + $N_t$ pilots &  0.0352&0.4380& 0.4897& 0.7281 & -2.0845 \\
    DPS-MIMO + $2N_t$ pilots &      0.0703&0.4597& 0.4600& 0.6763 & -3.8367 \\
    DJSCC-MIMO + $2N_t$ pilots &      0.0703 &\underline{0.5437} & \underline{0.4140} & \underline{0.5223} & -3.8341 \\
    BPG-LDPC + $2N_t$ pilots &  0.0703 &N/A & N/A & N/A & -3.8345 \\
    BPG-LDPC + $2N_t$ pilots &  0.2109 &0.4953 & 0.4411 & 0.6138 & -3.8339 \\
    DJSCC + Pro-BiG-AMP &  0.0272 &0.0513 & 0.9915 & 0.9944 & -0.0242 \\
    \bottomrule
\end{tabular}
\end{table}


The impact of MIMO channel dimensions on the recovery performance of our Blind-MIMOSC is shown in Fig. \ref{massiveMIMOFig}, where $T=24$, $N_u=1$, and SNR $=0$ dB. The horizontal axis denotes the number of transmit antennas $N_t$, and the number of receive antennas is $N_r=8N_t$. As $N_t$ increases, all performance metrics degrade. The reasons are two-fold: first, the number of channel parameters to be estimated, i.e., $N_rN_tK=1536N_t$, increases linearly with $N_t$, where channel estimation becomes increasingly challenging in higher dimensions; second, since the channel in the $k$-th block is estimated based on signal $\boldsymbol{X}_k\in\mathbb{C}^{N_t\times T}$, as $N_t$ grows while $T$ is fixed, the corresponding linear inverse problem of channel estimation becomes more ill‑conditioned according to linear algebra,
which in turn affects the performance of source recovery. To further validate this phenomenon, we compare Blind-MIMOSC against an oracle‑assisted LMMSE channel estimation benchmark that assumes perfect knowledge of $\boldsymbol{X}$ and uses it as the pilot for channel estimation. The channel estimation NMSE of Blind-MIMOSC remains within $0.7$ dB of the oracle bound, indicating that 1) Blind-MIMOSC with PVD operates near the optimal regime, and 2) the observed performance loss is an inherent consequence of increasing the MIMO size with coherence time $T$ fixed.

\begin{figure}[htb]
    \centering
    \includegraphics[width=0.78\linewidth]{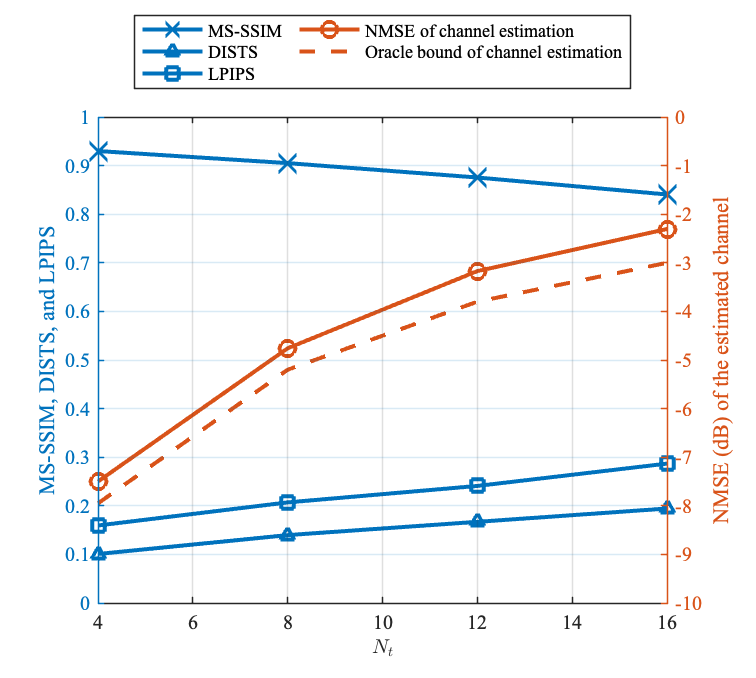}
    \caption{Impact of MIMO size on recovery performance with $T$ fixed.}
    \label{massiveMIMOFig}
\end{figure}

\subsection{CDL Channel Model}
\begin{figure*}[htb]
    \centering
    \setlength{\tabcolsep}{0.45pt} 
    \renewcommand{\arraystretch}{0} 
    \begin{tabular}{ccccccc}
        & 
        \multicolumn{1}{c}{\scriptsize \textbf{Blind-MIMOSC + }} & 
        \multicolumn{1}{c}{\scriptsize BPG-Capacity + } & 
        \multicolumn{1}{c}{\scriptsize BPG-LDPC + } & 
        \multicolumn{1}{c}{\scriptsize DJSCC-MIMO + } & 
        \multicolumn{1}{c}{\scriptsize DJSCC-MIMO + } & 
        \multicolumn{1}{c}{\scriptsize DPS-MIMO + } \\[2pt]

        \multicolumn{1}{c}{\scriptsize Ground truth} & 
        \multicolumn{1}{c}{\scriptsize \textbf{PVD}} & 
        \multicolumn{1}{c}{\scriptsize Perfect CSI}& 
        \multicolumn{1}{c}{\scriptsize $N_t$ pilots} & 
        \multicolumn{1}{c}{\scriptsize Perfect CSI} & 
        \multicolumn{1}{c}{\scriptsize $N_t$ pilots} & 
        \multicolumn{1}{c}{\scriptsize $N_t$ pilots} \\[2pt]
        
        \multicolumn{1}{c}{} & 
        \multicolumn{1}{c}{\scriptsize $\text{CBR}=0.023$} & 
        \multicolumn{1}{c}{\scriptsize $\text{CBR}=0.023$} & 
        \multicolumn{1}{c}{\scriptsize $\text{CBR}=0.211$ {\color{red}$\uparrow 800\%$}} & 
        \multicolumn{1}{c}{\scriptsize $\text{CBR}=0.023$} & 
        \multicolumn{1}{c}{\scriptsize $\text{CBR}=0.070$ {\color{red}$\uparrow 200\%$}} &
        \multicolumn{1}{c}{\scriptsize $\text{CBR}=0.070$ {\color{red}$\uparrow 200\%$}} \\[2pt]
        \includegraphics[width=0.141\linewidth]{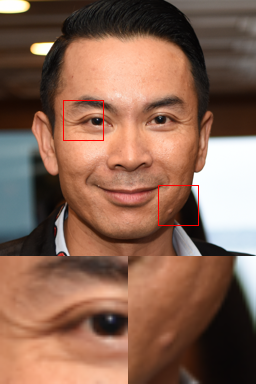} &
        \includegraphics[width=0.141\linewidth]{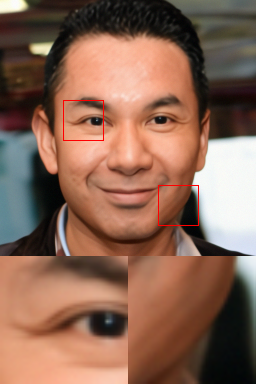} &
        \includegraphics[width=0.141\linewidth]{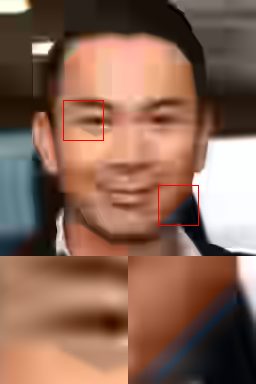} &
        \includegraphics[width=0.141\linewidth]{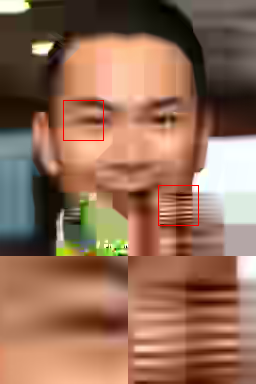} &
        \includegraphics[width=0.141\linewidth]{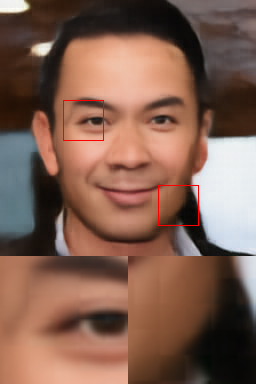} &
        \includegraphics[width=0.141\linewidth]{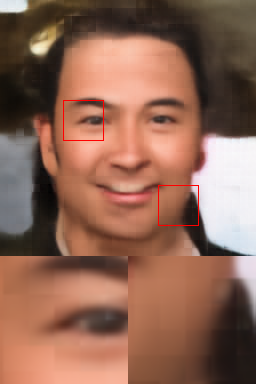} &
        \includegraphics[width=0.141\linewidth]{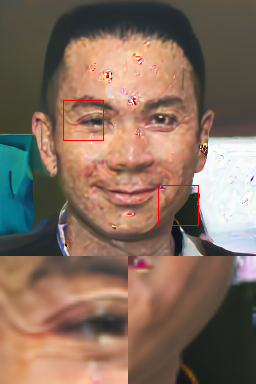}
    \end{tabular}
    \caption{Visual comparisons of image source under FFHQ-256 validation set and CDL-C dataset. Parameters are set as $N_r=128$, $N_u = 1$, $N_t = 16$, SNR $=0$ dB. The CDL-C channel remains constant for $T=24$ channel uses. 
    }
    \vspace{-0.3cm}
    \label{imagecompCDL}
\end{figure*}
In this experiment, the wireless channel is modeled using the clustered delay line (CDL) model defined in the 3GPP TR 38.901 standard for 5G New Radio (NR) systems.
The channel matrix for this experiment is generated using QuaDRiGa,
an open-source tool that implements the 3GPP CDL specifications. We use the CDL-C channel to evaluate the performance. 
Parameters are $N_r = 128$, $N_t = 16$, $N_u = 1$, $T=24$, and SNR $=0$ dB.  
The pilot-based channel estimation is given by the LMMSE estimation assuming the spatial covariance of the CDL-C channel is known.

The quantitative results on the FFHQ-256 validation set and the CDL-C channel dataset are listed in Table \ref{CDL}, and the visual comparisons of the recovered image source are shown in Fig. \ref{imagecompCDL}. It can be seen that Blind-MIMOSC with PVD maintains significant recovery performance with a drastically lower CBR.
We also test the case where the training and test sets of channels do not match, i.e., the channel scores are trained from the Rayleigh fading channel dataset. The results are shown in Table \ref{CDL} as Blind-MIMOSC + PVD$^*$ with two different CBRs, which maintains comparable source reconstruction quality to other methods in terms of DISTS and LPIPS, demonstrating the robustness of PVD. 
Regarding computational complexity, although PVD is faster than traditional diffusion models like DPS, it still takes approximately 20 seconds on our server to simultaneously recover the channel and the image source. Despite its non-negligible complexity, PVD represents a practical and promising solution for completely pilot-free channel-and-source recovery in MIMO communication systems.

\begin{table}[htb]
\setlength{\tabcolsep}{0.7pt}
\centering
\caption{Quantitative results on FFHQ-256 validation set and CDL channels. $N_r = 128$, $N_t = 16$, $N_u = 1$, $T=24$, SNR $=0$ dB. \textbf{Bold}: best \& \underline{Underline}: second-best among practical approaches.}
\label{CDL}
\begin{tabular}{l  |c  |c c c  |c }
    \toprule
    \multirow{3}{*}{Methods}  &\multirow{3}{*}{CBR}& \multirow{3}{*}{MS-SSIM$\uparrow$} & \multirow{3}{*}{DISTS$\downarrow$} & \multirow{3}{*}{LPIPS$\downarrow$} & Channel \\
     && & & & Estimation \\
     && & & & NMSE$\downarrow$ \\
    \midrule
    BPG-Capacity + Perfect CSI  &0.0234 & 0.8852 & 0.2853 & 0.4112 & N/A \\
    DJSCC-MIMO + Perfect CSI  &0.0234 & 0.9121 & 0.2314 & 0.2926 & N/A \\
    \midrule
    Blind-MIMOSC + PVD  &\textbf{0.0234} & \textbf{0.8870}& \textbf{0.1458}& \textbf{0.2353} & \textbf{-6.6141} \\
    Blind-MIMOSC + PVD$^*$ &0.0234& 0.5843 & 0.3546 & 0.4711 & -1.6494 \\
    Blind-MIMOSC + PVD$^*$ &0.0703& 0.7715 & \underline{0.2256} & \underline{0.3388} & -2.6216 \\
    DPS-MIMO + $N_t$ pilots  &0.0703& \underline{0.8392} & 0.3629 & 0.5620 & -4.2519 \\
    DJSCC-MIMO + $N_t$ pilots  &    0.0703& 0.7524 & 0.3256 & 0.4595 & -4.2555 \\
    BPG-LDPC + $N_t$ pilots  &0.0703& N/A & N/A & N/A & -4.2571 \\
    BPG-LDPC + $N_t$ pilots  &    0.2109& 0.6513 & 0.3913 & 0.5937 & -4.2566 \\
    DJSCC + Pro-BiG-AMP  &0.0256& 0.0512 & 0.9921 & 0.9915 & -0.0084 \\
    \bottomrule
\end{tabular}
\end{table}

\section{Conclusions}\label{conclusion}
In this paper, we presented a novel blind MIMO semantic communication framework, named Blind-MIMOSC, with a DJSCC transmitter and a diffusion-based blind receiver. We proposed the PVD model for the blind receiver to leverage the structural characteristics of the channel and the source data via diffusion models, thereby enabling simultaneous channel estimation and source recovery in block-fading MIMO systems without using any pilots. Extensive simulations under different MIMO channel conditions showed that the proposed framework significantly improves channel estimation accuracy, source recovery quality, and transmission efficiency compared to state-of-the-art approaches.
These results underscore the potential of the proposed Blind-MIMOSC framework with the PVD model to enhance the overall system performance in MIMO communications, and open new avenues for the application of generative AI in wireless communications.

There are still several promising research directions worthy of further exploration. For instance, how to significantly reduce the computational complexity of diffusion models during the inference phase is a critical issue that urgently needs to be addressed. In this direction, further research into more efficient generative models, such as flow matching, could be conducted to enhance the practicality of our Blind-MIMOSC framework.


\appendices
\section{Pilot-Based Approaches via Conditional Entropy Minimization}\label{pilotbased}
The conditional entropy $H\left(\boldsymbol{H}_0,\boldsymbol{D}_0\vert \boldsymbol{Y}\right)$ can be used as the objective function for pilot-based transceiver design to obtain the channel estimator $g_{\boldsymbol{\theta}_1}$ parameterized by $\boldsymbol{\theta}_1$ and JSCC codec $\{f_{\boldsymbol{\gamma}},g_{\boldsymbol{\theta}_2}\}$ parameterized by $\boldsymbol{\gamma}$ and $\boldsymbol{\theta}_2$. For instance, let us insert pilot symbols, denoted by $\boldsymbol{X}_p$, into transmitted signals, i.e., $\boldsymbol{X} = \left[\boldsymbol{X}_p, \boldsymbol{X}_d\right]$ where $\boldsymbol{X}_d$ is the data signal encoded from source $\boldsymbol{D}_0$. The received signal is $\boldsymbol{Y} = \left[\boldsymbol{Y}_p, \boldsymbol{Y}_d\right]$ with $\boldsymbol{Y}_p$ and $\boldsymbol{Y}_d$ being the received pilot signal and data signal, respectively. The conditional entropy is relaxed as
\begin{subequations}
\begin{align}
    &H\left(\boldsymbol{H}_0,\boldsymbol{D}_0\vert  \boldsymbol{Y}\right) \nonumber\\
    ={}& H\left(\boldsymbol{H}_0, \boldsymbol{D}_0\vert\boldsymbol{Y}_d,\boldsymbol{Y}_p,\boldsymbol{X}_p\right)\label{eq2:step1}\\ 
    \leq{}& H\left(\boldsymbol{H}_0, \boldsymbol{D}_0,\boldsymbol{X}\vert\boldsymbol{Y}_d,\boldsymbol{Y}_p,\boldsymbol{X}_p\right)\label{eq2:step2}\\ 
    ={}& H\left(\boldsymbol{H}_0\vert\boldsymbol{Y}_d,\boldsymbol{Y}_p,\boldsymbol{X}_p\right)\nonumber \\
    &+ H\left(\boldsymbol{D}_0,\boldsymbol{X}\vert\boldsymbol{Y}_d,\boldsymbol{Y}_p,\boldsymbol{X}_p,\boldsymbol{H}_0\right)\label{eq2:step3}\\
    \leq{}& H\left(\boldsymbol{H}_0\vert\boldsymbol{Y}_p,\boldsymbol{X}_p\right) + H\left(\boldsymbol{D}_0,\boldsymbol{X}_d\vert\boldsymbol{Y}_d,\boldsymbol{H}_0\right)\label{eq2:step4}\\
    \leq{}& \underbrace{\mathbb{E}_{p\left(\boldsymbol{H}_0,\boldsymbol{Y}_p,\boldsymbol{X}_p\right)}\left[-\ln q_{\boldsymbol{\theta}_1}\left(\boldsymbol{H}_0\vert\boldsymbol{Y}_p,\boldsymbol{X}_p\right)\right]}_{\text{Loss function of the channel estimator $g_{\boldsymbol{\theta}_1}$}}\nonumber\\
    &+\underbrace{\mathbb{E}_{p_{\boldsymbol{\gamma}}\left(\boldsymbol{D}_0,\boldsymbol{X}_d,\boldsymbol{Y}_d,\boldsymbol{H}_0\right)}\left[-\ln q_{\boldsymbol{\theta}_2}(\boldsymbol{D}_0,\boldsymbol{X}_d\vert\boldsymbol{Y}_d,\boldsymbol{H}_0)\right]}_{\text{Loss function of the JSCC codec $\{f_{\boldsymbol{\gamma}},g_{\boldsymbol{\theta}_2}\}$}},\label{eq2:step8}
\end{align}
\end{subequations}
where $q_{\boldsymbol{\theta}_1}(\boldsymbol{H}_0 \vert\boldsymbol{Y}_p,\boldsymbol{X}_p)$ and $q_{\boldsymbol{\theta}_2}(\boldsymbol{D}_0,\boldsymbol{X}_d\vert\boldsymbol{Y}_d,\boldsymbol{H}_0)$ are the variational approximations of the true \textit{a posteriori} distributions $p(\boldsymbol{H}_0 \vert\boldsymbol{Y}_p,\boldsymbol{X}_p)$ and $p_{\boldsymbol{\gamma}}(\boldsymbol{D}_0,\boldsymbol{X}_d\vert\boldsymbol{Y}_d,\boldsymbol{H}_0)$, respectively. The detailed explanations for each step are as follows. The equality in \eqref{eq2:step1} holds because $\boldsymbol{X}_p$ is deterministic given $\boldsymbol{Y}_p$. The inequality in \eqref{eq2:step2} is valid since introducing the intermediate variable $\boldsymbol{X}$ increases the joint conditional entropy. The equality in \eqref{eq2:step3} results from the chain rule of conditional entropy. The inequality in \eqref{eq2:step4} follows from the property that dropping off condition variables increases the uncertainty. 
Finally, the inequality in \eqref{eq2:step8} introduces two cross-entropies as the upper bounds of the two terms in \eqref{eq2:step4}, respectively. 
Thus, pilot-based two-stage approaches are obtained by minimizing the two cross-entropies sequentially, where the first cross-entropy is minimized to achieve the channel estimator, and the second cross-entropy is minimized to obtain the JSCC codec.
This separated optimization reduces the blind channel-and-source recovery scheme to pilot-based two-stage schemes, which is generally suboptimal due to the data processing inequality \cite{tse2005fundamentals} and the additional assumptions imposed on $\boldsymbol{H}_0$, $\boldsymbol{X}$, and $\boldsymbol{D}_0$ to achieve channel estimation and decoding. Pilot-based two-stage approaches such as the DPS-MIMO algorithm \cite{chung2023diffusion} and the DJSCC-MIMO algorithm \cite{wu2024deep} are included as baselines in our numerical comparisons in Section~\ref{NR}.

\bibliographystyle{IEEEtran}
\bibliography{Blind-MIMOSC.bib}

\end{document}